\documentclass[numbook, envcountsect, envcountsame, envcountreset, runningheads, smallextended]{svjour3}

\usepackage{graphicx}        % standard LaTeX graphics tool
\usepackage{multicol}        % used for the two-column index
\usepackage[ruled,vlined]{algorithm2e}
\usepackage{amsmath}%,amsthm}
\usepackage{dsfont}
\usepackage{bm}
\usepackage{animate}
\usepackage{color}
\usepackage{placeins}
\usepackage[hidelinks]{hyperref}
\usepackage{pbox}

\usepackage{ mathrsfs }
\usepackage{multirow}
\usepackage[table]{xcolor}

							\newcommand{\R}{\mathds{R}}
							
							\newcommand{\Prob}{\mathds{P}}
							\newcommand{\Exp}{\mathds{E}}
							\newcommand{\one}{\mathds{1}}

							\newcommand{\X}{\textbf{X}}
							\newcommand{\x}{\textbf{x}}
							\newcommand{\Ubold}{\textbf{U}}
							\newcommand{\ubold}{\textbf{u}}
							
							\definecolor{light-gray}{gray}{0.90}
							\definecolor{mid-gray}{gray}{0.50}
\newtheorem{ass}[theorem]{Assumptions}
\begin{document}
\rowcolors{2}{gray!15}{white}
\title{Sequential Monte Carlo Samplers for capital allocation under copula-dependent risk models}
\titlerunning{SMC Samplers for capital allocation under copula-dependent risk models}

\author{Rodrigo S. Targino \and Gareth W. Peters \and Pavel V. Shevchenko}
\institute{Rodrigo S. Targino \at University College London, London, UK,\\
					\email{r.targino.12@ucl.ac.uk}
					\and Gareth W. Peters \at University College London, London, UK,
					\and Pavel V. Shevchenko \at CSIRO, Sydney, Australia and University College London, London, UK}
%\journalname{Insurance: Mathematics and Economics}
\date{\today}
\maketitle

\begin{abstract}
In this paper we assume a multivariate risk model has been developed for a portfolio and its capital derived as a homogeneous risk measure. The Euler (or gradient) principle, then, states that the capital to be allocated to each component of the portfolio has to be calculated as an expectation conditional to a rare event, which can be challenging to evaluate in practice. We exploit the copula-dependence within the portfolio risks to design a Sequential Monte Carlo Samplers based estimate to the marginal conditional expectations involved in the problem, showing its efficiency through a series of computational examples.
%
%In practical capital allocation settings, one should have the flexibility to model the individual risk processes in different divisions or business units of the financial institution with flexible dependence structures between the loss processes. In such cases, axioms of coherent capital allocation result in allocation rules for capital risk measures such as Value at Risk and Expected Shortfall which are challenging to evaluate in practice. In this paper we assume that a risk model has been developed and capital derived. Given this starting point, we develop a new class of efficient rare event sampling techniques to address the capital allocation challenge.
\end{abstract}

\vspace{3mm}
\hspace{-5mm}\textbf{MSC 2010:} 65C40, 65C05, 62P05 

%================================================================================================
%================================================================================================
\section{Introduction}
\label{sec:intro}
%================================================================================================
%================================================================================================

Since financial institutions are in the business of managing and reallocating risks, an important part of their internal risk management is to have an appropriate level of capital as a buffer against unexpected losses. Typically, retail banks, investment banks and insurance companies must satisfy their local jurisdiction version of capital adequacy, which are usually specified by the local regulatory authorities according to some version of either the Basel II/III banking supervision guides or the Solvency II insurance guides.

In this context capital may refer to two different quantities: Economic Capital (EC) or Regulatory Capital (RC). The first is the capital that would have be chosen in the absence of regulation. It represents the amount the institution estimates in order to remain solvent at a given confidence level and fixed time horizon and is set in order to meet some credit risk rating. The Regulatory Capital, in turn, reflects the needs given by regulatory guidance and rules. It is important to note that the capital actually been held by the institution (henceforth referred to as \emph{capital}) will always be the maximum between the EC and the RC.

Both for banks and insurance companies, regulation has evolved towards Regulatory Capital based on risk measures (see Section \ref{sec:riskContributions}). For banks, the Basel II accord set the standard to the Value at Risk (still in use at Basel III) while insurance directives such as Solvency II and the Swiss Solvency Test diverge on the risk measure to be used (the first suggests the usage of Value at Risk and the former Expected Shortfall). It is important to note that although the RC and the EC will differ in most of the financial institutions, both quantities are usually based on the same class of risk measures, differing only on the confidence level.

This work is focused on studying the problem of capital allocation, with particular focus on Operational Risk (OpRisk) capital -- see \cite{shevchenko2011modelling} for an text-book introduction to OpRisk modelling. In order to solely study this aspect, we will assume that the parametric risk models have been selected and the parameter estimations performed in each business unit or division for all the relevant risk types. Then using these models the bank or insurance company has obtained an estimate of the total capital from the model. Based on this capital figure, our work aims to study a problem that follows the calculation of the capital to be held, namely the second order problem of the capital allocation to different divisions and business units as a capital charge (see Table \ref{tbl:businessLinesEventType}). Once the overall capital is calculated, the financial institution faces the problem of how to allocate this given capital among different risk sources, in order to understand how much each risk cell contributes to the total risk (capital) and in order to asses their risk management controls and performance, as part of the process discussed in the Pillar III of Basel II/III. Put another way there is a capital charge that must be allocated to each division and business unit which must reflect commensurately the risk profile of the given business unit. This continues to provide an incentive for banks following the Advanced Measurement Approach (AMA) to carefully model their dependence structures.

%Therefore, in the context of the study of capital allocation, we begin from the premise that we have access to the Operational Risk (Regulatory / Economic) capital of a bank and we want to have it allocated either to different business lines or event types, following the Basel II Business Line x Event Type matrix as in Table \ref{tbl:businessLinesEventType}. For an text-book introduction to Operational Risk (OpRisk) modelling, see \cite{shevchenko2011modelling}.

\begin{table}
\centering
\begin{tabular}{ll}
	\hline
		& \textbf{Business line} \\
	\hline
	1 & Corporate finance \\ 
	2 & Trading and sales \\
	3 & Retail banking \\ 
	4 & Commercial banking \\
	5 & Payment and settlement \\
	6 & Agency services \\
	7 & Asset management \\
	8 & Retail brokerage \\
	\hline
\end{tabular}
\hspace{0.5cm}
\begin{tabular}{ll}
\hline
\rowcolor{white}
	& \textbf{Event type} \\
	\hline
1 & Internal fraud \\
2 & External fraud \\
3 & Employment practices and workplace safety \\
4 & Clients, products and business practices \\
5 & Damage to physical assets \\
6 & Business disruption and system failures \\
7 & Execution, delivery and process management \\
\hline
\end{tabular}
\caption{Basel II business lines (left) and event types (right) -- see \cite{baselii}, Annexes 8 and 9.}
\label{tbl:businessLinesEventType}
\end{table}

Apart from the fact that losses in some of the these risk cells may be dependent, the Basel II accord \cite{baselii} $\mathsection 657$ ensures that the capital estimate can have diversification benefits if dependency modelling is approved by the local regulator. In other words, the bank may be authorized to set aside less Regulatory Capital if they can demonstrate evidence for dependence features in their loss processes between each business line or between risk types within a business line.

We will also assume the dependency among all risk cells in the portfolio is known from the first phase of model selection and estimation. More precisely, we will assume that the bank's portfolio consists of $d$ individual losses (in a risk cell level) denoted by $X_1,...,X_d$, each one modelled as random variables (rv's) with continuous cumulative distribution function (cdf) given by $F_i$, $i=1,...,d$.

The dependence structure of the losses will be given by a (known) copula $C(u_1,...,u_d)$ (see Appendix \ref{sec:gaussianCopula} for the definition and some results regarding copulas), leading to a joint distribution of the losses given by
$$F_\X ( \x ) = C\big( F_1(x_1),...,F_d(x_d)\big),$$
where $\X = (X_1,...,X_d)$, and $\x = (x_1,...,x_d)$.

In the recent years many academic works were devoted to the joint modelling of operational losses and it impact on capital calculation. Recently, \cite{brechmann2014flexible} introduced a zero-inflated dependence model, which is then coupled using different copulas (Archimedean, elliptical, individual Student's t and vine). Previously, \cite{Giacometti08aggregationissues} used $\alpha$-stable marginal distributions and Student's t copulas (both symmetric and skewed ones); \cite{bocker2008modelling} derived approximations for the operational Value at Risk assuming generalized Pareto distributions and L\'evy copulas. In the context of OpRisk, \cite{peters2009dynamic} developed a dynamic OpRisk model with copula dependence structures between the frequency, severity and annual loss as possible model structures. In addition, a detailed account of dependence modelling and capital estimation can be found in \cite{OpRiskBOOK1} and \cite{OpRiskBOOK2}.

The remainder of the paper is structured as follows. In Section \ref{sec:riskContributions} we present results related to risk contributions, and special attention is drawn to the Euler allocation rule and its duality with expectations conditional to rare events. After understanding the relationship between the Euler principle and conditional expectations, Section \ref{sec:copSMC} presents different methods to estimate these expectations: from a simple Monte Carlo scheme to recently developed Importance Sampling approaches. To be able to discuss the Sequential Monte Carlo Sampler (SMCS) algorithm proposed, Section \ref{sec:intermediateSets} briefly discusses how to design a sequence of probability densities converging to the conditional distribution involved in the allocation problem. The formal sequential procedure is, then, carefully described in Section \ref{sec:SMC} and the algorithm's ingredients necessary for the allocation problem designed in Section \ref{sec:linearConstraints}. We conclude the paper with some simulation examples on Section \ref{sec:Examples} and final remarks on Section \ref{sec:conclusions}.

%================================================================================================
%================================================================================================
\section{Risk contributions and capital allocations}
\label{sec:riskContributions}
%================================================================================================
%================================================================================================

This section recalls a few key results that have been developed primarily relating to coherent capital allocation principles. A brief overview is provided for the relevant theory of risk contributions, with a focus on Euler allocation principles, for which a more detailed introduction is provided in \cite{mcneil2010quantitative}, Section 6.3, \cite{tasche2008capital} and \cite{rosen2010risk}, Section 3.

In general one may consider the notation in which $X_1,...,X_d$ denotes the returns of $d$ different assets in a portfolio. In the case of OpRisk modelling, this notation will correspond to losses from $d$ different business unit and risk type combinations within divisions of a banking or insurance institution. 

If the weights of these assets in the portfolio are given by $\bm{\lambda} = (\lambda_1,...,\lambda_d) \in \Lambda \subset \R^d \setminus \{\textbf{0}\}$ then we will denote the portfolio-wide loss by
\begin{equation}
\label{eq:weightedPortfolio}
X(\bm{\lambda}) = \sum_{i=1}^d \lambda_i X_i.
\end{equation}
In particular, if $\bm{\lambda} = (1,...,1)$ we will write $X = X(\bm{\lambda})$. In the case of OpRisk modelling this aggregate loss amount given by $\sum_{i=1}^d \lambda_i X_i$ would represent the institution-wide total annual loss and typically the weights would be unity.

If all the losses $X_1,...,X_d$ are defined in a common probability space $(\Omega, \mathcal{F}, \Prob)$, then we denote $\rho \, : \, \Omega \rightarrow \R$ a generic risk measure. Given a specific risk measure $\rho$, the function $r_\rho \, : \, \Lambda \rightarrow \R $ such that $r_\rho(\bm{\lambda}) = \rho(X(\bm{\lambda}))$ will be called risk-measure function.

In the context of OpRisk, the Basel II/III guidelines specify clearly the recommended regulatory requirements for capital adequacy standards. For instance in the European Union (EU) region the Capital Requirements Directives (CRD) for the financial services industry have developed a supervisory framework to reflect its guidelines on capital measurement and capital standards. In the EU region the CRD-IV package entered into force on 17 July 2013 and reflects the Basel III standards expanding existing Basel II regulatory EU directives (2013/36/EU) and EU regulation (575/2013). These EU directives, like their counterparts in other jurisdictions include specifications on the amounts of Tier I and Tier II capital, liquidity ratios and other relevant capital adequacy criteria. In this paper we are not so interested in specific break up of capital components, instead in this regard, throughout this paper we will refer to capital rather loosely as corresponding to generically the total amount of assets that a financial institution must hold to mitigate their yearly loss exposure obtained from quantification of a risk measure based on the institutions OpRisk models.

From now on, let us assume the bank's capital is defined as a function of a given risk measure $\rho$, defined as either the standard deviation, the Value at Risk or the Expected Shortfall (see Definition \ref{def:riskMeasures}). For example, for regulatory purposes the OpRisk capital to be held is given by $\rho(X(\bm{\lambda})) = VaR_\alpha(X(\bm{\lambda}))$, with $\alpha = 0.999$ for the Regulatory Capital and $\alpha = 0.9995$ for the Economic Capital, for example.

\begin{definition}[Particular choices of risk measures] \label{def:riskMeasures} If $X_1,...,X_d$ are continuous random variables, and $X=\sum_{i=1}^d X_i$ three of the most popular choices of risk measures are given by
\begin{enumerate}
	\item Standard deviation: $\rho(X) = \sqrt{Var(X)}$;
	\item Value at Risk: $\rho(X) = VaR_\alpha(X) := \inf\{x \in \R \, : \, F_X(x) = \alpha\}$;
	\item  Expected Shortfall: $\rho(X) = ES_\alpha(X) := \Exp[X \, | \, X \geq VaR_\alpha(X)]$.
\end{enumerate}
\end{definition}

Assuming the risk measure $\rho$ has been chosen and that $\rho(X(\bm{\lambda}))$, the risk (capital) of the portfolio (institution) has already been calculated, the allocation process consists of understanding how much of the risk (capital) is due to each of the constituents of the portfolio. In the case of OpRisk, this involves understanding what each divisions total capital requirement would be across all the Basel III risk types, as well as what each individual business units and risk types combinations capital requirement should be, based on a ``fair'' risk based capital allocation from the institutions total capital requirement. This involves the disaggregation of the total institutions capital back down the business unit/risk type structure that will be specific for a given institution but can be generically represented for instance by the 56 business unit/risk type categories proposed in Basel II/III.

More formally, let us denote by $\mathscr{C}_i^{\rho}(\bm{\lambda})$ the capital allocated to one unit of $X_i$ when the portfolio's loss is given by $X(\bm{\lambda})$. For the sake of simplicity, to derive the Euler allocation we will accept the following set of assumptions.

\begin{ass} If the individual and portfolio losses are given, respectively by $X_1,...,X_d$ and Equation (\ref{eq:weightedPortfolio}) then we assume that
\begin{itemize}
	\item[(i)] the capital allocated to the position $\lambda_i X_i$ is given by $\lambda_i \mathscr{C}_i^{\rho}(\bm{\lambda})$;
	\item[(ii)] the overall risk capital $r_\rho(\bm{\lambda})$ is fully allocated to the individual positions in the portfolio: \begin{equation}
	\sum_{i=1}^d \lambda_i \mathscr{C}_i^{\rho}(\bm{\lambda}) = r_\rho(\bm{\lambda}).
\label{eq:perUnityCapitalAllocation}
\end{equation}
\end{itemize}
\end{ass}

The interpretation of the first assumption is that we require proportional positions to have proportional capital shares and the second one ensures the total capital will be allocated to the individual positions. As a result of these assumptions, we only need to calculate $\mathscr{C}_i^{\rho}(\bm{\lambda})$, for $i=1,...,d$.

Following the nomenclature in \cite{mcneil2010quantitative}, we will say $\mathscr{C}_i^{\rho} \, : \, \Lambda \rightarrow  \R^d$  is a per-unit capital allocation principle if (\ref{eq:perUnityCapitalAllocation}) is satisfied for all $\bm{\lambda} \in \Lambda$.

So far we have not imposed or assumed any specific requirement on the risk measure used as a base for the allocation principle, but to introduce the popular Euler allocation principle we need to restrict ourselves to the class of positive homogeneous (of degree one) risk measures.

\begin{definition}[Positive homogeneity] A risk measure $\rho$ is said to be positive homogeneous (of degree one) if $\rho(\lambda X) = \lambda \rho(X),$ for any random variable $X$ and $\lambda > 0$.
\end{definition}

It is straightforward to verify that Value at Risk (VaR) is positive homogeneous, as well as any coherent risk measure (in the sense of \cite{artzner1999coherent}). For positive homogeneous risk measures it is also trivial to show that the associated risk measure function $r_\rho$ satisfies $r_\rho(t \bm{\lambda}) = t r_\rho(\bm{\lambda})$. Therefore, applying Euler's homogeneous function theorem (see Appendix \ref{sec:eulerThm}) on $r_\rho$, we have that
\begin{equation}
r_\rho(\bm{\lambda}) = \sum_{i=1}^d \lambda_i \frac{\partial r_\rho}{\partial \lambda_i}(\bm{\lambda})
\label{eq:eulerPositiveHomogeneous}
\end{equation}

The combination of (\ref{eq:perUnityCapitalAllocation}) and (\ref{eq:eulerPositiveHomogeneous}) leads to the so-called Euler allocation principle (sometimes referred as allocation by the gradient), where the capital allocated to the i-th component of the portfolio is given by the partial derivatives, with $\mathscr{C}_i^{\rho}(\bm{\lambda}) := \frac{\partial r_\rho}{\partial \lambda_i}(\bm{\lambda})$.

The Euler allocation principle arises in different contexts in the literature. For example, in \cite{denault2001coherent} and \cite{kalkbrener2005axiomatic} the Euler principle is motivated by two (different) sets of axioms, leading to coherent allocation principles (for a relationship between coherent risk measures and coherent capital allocations see \cite{buch2008coherent}). 

Assuming that $X_1,...,X_d$ are continuous random variables at the point at which the risk measure is evaluated, we now present some explicit forms of the Euler contributions, based on the different risk measures presented in \ref{def:riskMeasures}.

\begin{proposition} 
\label{prop:eulerAllocation} If $X = \sum_{i=1}^d X_i$ and $\textbf{X} = (X_1,...,X_d)$ has a joint continuous density, then the Euler allocation takes the following form
\begin{enumerate}
\item Standard deviation: $\displaystyle \rho(X) = \sqrt{Var(X)} \Longrightarrow \mathscr{C}^{\sigma}_i(X_i) = \frac{Cov(X_i, \, X)}{\sqrt{Var(X)}}$;
\item Value at Risk: $\rho(X) = VaR_\alpha(X) \Longrightarrow \mathscr{C}^{VaR}_i(X_i) = \Exp[X_i \, | \, X = VaR_\alpha(X) ]$;
\item Expected Shortfall: $\rho(X) = ES_\alpha(X) \Longrightarrow \mathscr{C}^{ES}_i(X_i) = \Exp[X_i \, | \, X \geq VaR_\alpha(X) ]$.
\end{enumerate}
\end{proposition}

\begin{proof} See \cite{mcneil2010quantitative}, Section 6.3 and references therein.
\end{proof}

\begin{remark}
Proposition \ref{prop:eulerAllocation} is still valid even if the distribution of $\textbf{X}$ is not continuous but other technical conditions should be satisfied (see \cite{tasche1999risk} and \cite{gourieroux2000sensitivity}).
\end{remark}

%-----------------------------------------------------------------------------------------------------------------------
\subsection{\textbf{Euler allocation in a hierarchical structure}} \label{sec:hierarchical_allocation}
Here we briefly extend the concept of Euler allocations to a Bank structure divided, for example, in Business Units and Event types, as in Table \ref{tbl:businessLinesEventType}.

\begin{figure}[h]
\begin{center}
\includegraphics[width=0.6\textwidth]{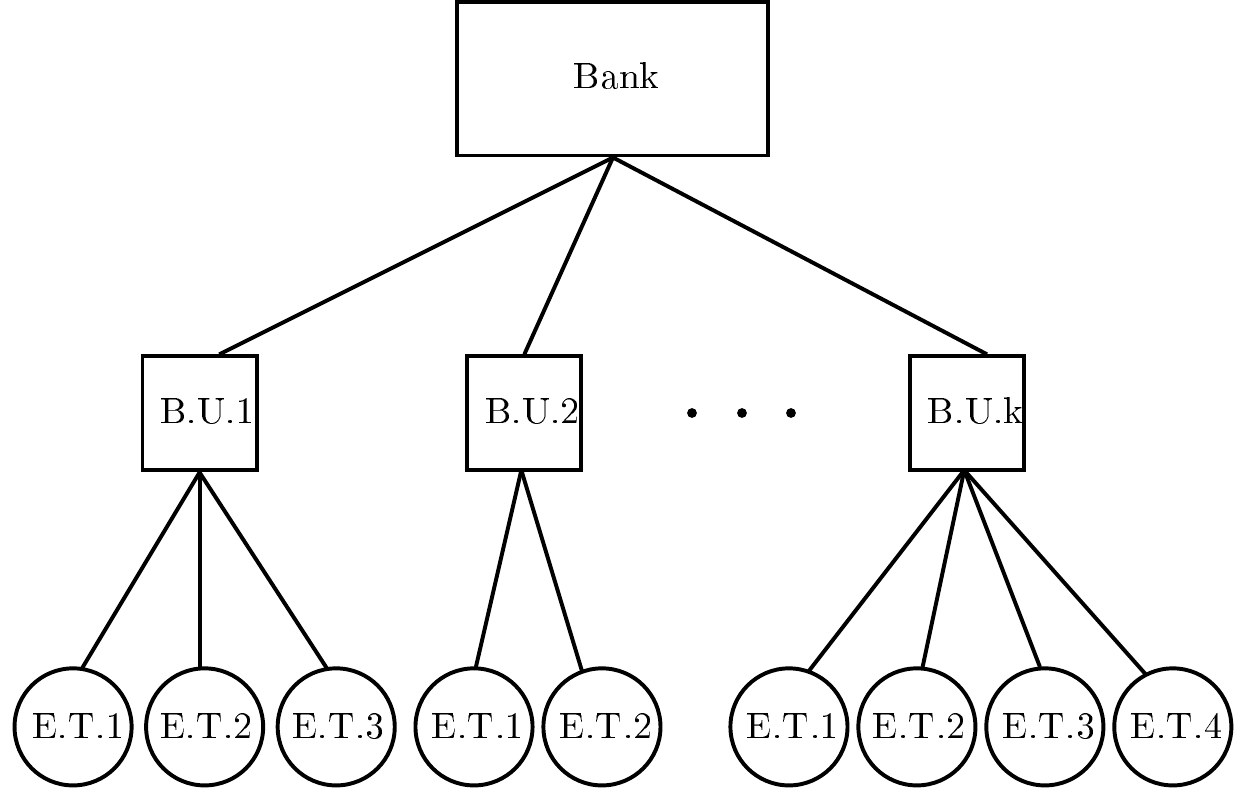}
\end{center}
\caption{Hierarchical bank structure, with $k$ B.U.'s.}
\label{fig:hierarchical_bank}
\end{figure}

Let us assume a bank has a structure given as in Figure \ref{fig:hierarchical_bank}, comprising of $K$ Business Units (B.U.'s) and $d_l$ Event Types (E.T's) in each of its B.U.'s ($l=1,...,K$). In this context we define $d=\sum_{l=1}^K d_l$ as the total number of cells for which capital should be allocated, $X = \sum_{i=1}^d X_i$ the bank loss and $X_{[l]} = \sum_{m=1}^{d_l} X_{m}$ for $l=1,...,K$ the loss in the $l$-th B.U..

Assuming the bank capital is given by $ES_\alpha(X)$ (as in Definition \ref{def:riskMeasures}) then the Euler principle states that the capital at the B.U. level should be given by $\Exp[X_{[l]} \, | \, X > VaR_\alpha(X)]$ for each B.U. $l=1,...,K$. To allocate the capital calculated for the $l$-th unit to its E.T.'s we can assume this capital is a (homogeneous) risk measure defined as
\begin{equation*}
\rho(X_{[l]}) = \Exp[X_{[l]} \, | \, X > VaR_\alpha(X)].
\label{eq:hierarchical_rho}
\end{equation*}
Then, following the Euler principle it is easy to check the allocation for the $m$-th E.T. in the $l$-th B.U. is given by
$$\Exp[X_{m} \, | \, X > VaR_\alpha(X)].$$
The reader should note that the same allocations could have been derived ``heuristically" in the following way. First the total capital $ES_\alpha(X)$ is allocated directly to each of the $d$ E.T.'s. Then for the $l$-th B.U. the capital is computed as the sum of the allocations at the $d_l$ E.T.'s in this unit. Although the result would be the same, we emphasize the first method, as it is a direct application of the Euler principle (twice).

\section{From capital allocation to conditional expectations} \label{sec:copSMC} 
%================================================================================================
%================================================================================================

In this section we explore the fact that, using the two most important risk measures for OpRisk, the capital allocation framework just described under the Euler allocation principle can be redefined as the calculation of conditional expectations. This is practically a very appealing result as it means that the allocation of the capital can be estimated using specialized Monte Carlo sampling solutions, presented in detail below.

Throughout it will be assumed that the marginal loss processes are continuous random variables. Therefore, one has that all the marginal inverse distribution functions (quantile functions) $F_i^{-1}$ are well defined and are also continuous. Moreover, one may apply Sklar's theorem to state that the dependence structure of the loss vector $\X$ will be uniquely determined by a copula function $C$ (see Appendix \ref{sec:gaussianCopula} for some results on copula theory). Formally the joint cdf and pdf of the vector $\X$ can be written as
$$F_\X(\x) = \Prob[X_1 \leq x_1, ..., X_d \leq x_d] = C(F_1(x_1), ..., F_d(x_d)) \text{ and}$$
$$f_\X(\x) = c\big( F_1(x_1),...,F_d(x_d)\big) \prod_{i=1}^d{f_i(x_i)},$$
where $C$ and $c$ are the copula and copula density, respectively.

From Proposition \ref{prop:eulerAllocation} we can see both the allocation based on VaR and on ES can be calculated as an expectation of the form
\begin{equation}
C^{\rho}_i(X_i) = \Exp[h(\X) \,| \, g(\X) \in A],
\label{eq:conditionalExpectation}
\end{equation}
for the following choices of $h$, $g$ and $A$:
\begin{enumerate}
	\item For VaR: $h(\X) = X_i, \ g(\X) = \sum_{i=1}^d X_i, \ A=\left[ VaR_\alpha(\sum_{i=1}^d X_i ), \, VaR_\alpha(\sum_{i=1}^d X_i ) \right]$;
	\item For ES: $h(\X) = X_i, \ g(\X) = \sum_{i=1}^d X_i, \ A=\left[ VaR_\alpha(\sum_{i=1}^d X_i ), \, + \infty \right)$.
\end{enumerate}

Note that in any practical scenario $\alpha$ will be chosen to be close to one and both conditioning events will have very small probabilities, ie, $\Prob[g(\X) \in A] \approx 0 $. More precisely, the event will have probability $1-\alpha$ for the ES case and probability zero (regardless of the choice of $\alpha$) in the VaR case, since the set $A$ is just a point. Therefore in practice, in the VaR case one would work instead with an $\epsilon$ approximation by setting the conditioning event as the $\epsilon$-ball of $A$ given by $\mathcal{B}_{\epsilon}(A)$ specified according to $\mathcal{B}_\epsilon(A) =\left[ VaR_\alpha(\sum_{i=1}^d X_i ) - \epsilon, \, VaR_\alpha(\sum_{i=1}^d X_i )+\epsilon \right]$ for some small positive $\epsilon$ in the neighbourhood of zero, $\epsilon \in n.e.(0^+)$. This type of approach was advocated in \cite{glasserman2005measuring} and \cite{del2013introduction}.

From a risk management perspective, many other choices of $h$ and $g$ can be of interest. For example, if one is interested in measuring the impact of marginal tail events in the portfolio, one could calculate expectations of the form $\Exp[ \sum_{i=1}^d X_i \, | \, X_k > VaR_\alpha(X_k) ]$, where the choices of $h$, $g$ and $A$ are trivial.

In practice, apart from very particular choices of dependence structure and marginal distributions (eg, \cite{asimit2013evaluating}), analytic representations for the expectation in (\ref{eq:conditionalExpectation}) will not be available and we will need to resort in some numerical approximation to evaluate such quantities in practice. There are several classes of algorithm of relevance to undertake this task so, before presenting our Sequential Monte Carlo (SMC) Samplers based solution, we review some of the most recent proposals in the literature.

%================================================================================================
\subsection{\textbf{Simulation methods to calculate conditional expectations of rare-events}} \label{sec:simForConditionalExp}
%================================================================================================
In order to fix the notation to be used throughout this work, in this section we describe a generic Monte Carlo based estimator, before detailing a few specialized approaches to be compared with our proposed method. It is important to recognize the difference between Monte Carlo, Importance Sampling and Sequential Monte Carlo Sampler based solutions. These are different categories of solution technique as will be detailed below.

The generic expectation $\Exp[h(\X) \, | \, g(\X) \in A]$ can be approximated via a Monte Carlo simulation through the use of a set of $N$ weighted samples $\{ \x^{(j)}, \, W^{(j)} \}_{j=1}^N$, with $\sum_{j=1}^N W^{(j)}=1$, from the conditional distribution $f_\X(\x \, | \, g(\x) \in A)$. The approximation for the expectation is, then, given by
$$\Exp[h(\X) \, | \, g(\X) \in A] \approx \sum_{j=1}^N W^{(j)} h(\x^{(j)}).$$

For example, if we can directly sample i.i.d. realizations from the distribution of $\X \, | \, g(\X) \in A$ (eg, using a rejection scheme) we would have $W^{(j)} = 1/N$. In general, though, samples from the above conditional distribution are not readily available and we need to use an importance sampling distribution (coupled with a rejection step), calculating the weights accordingly to remove bias. In both cases a rejection mechanism can be used, accepting only those particles satisfying the conditioning event. However, if done naively this would result in very large rejection rates that would behave poorly as the conditioning event becomes rarer or the dimension $d$ gets large.

In this setup, the conditioning event we are interested in can be defined as $\left[ \X \in \mathcal{G}_\X \right],$ where
\begin{equation} \label{eq:G_Z}
	\mathcal{G}_\X : = \{ \x \in \R^d \, : \, g(\x) \in A \},
\end{equation}
since $g(\X) \in A \Longleftrightarrow \X \in \mathcal{G}_\X$. Given that our multivariate loss model is uniquely characterized by a copula (either explicitly or implicitly, through a parametric joint distribution) the region $\mathcal{G}_\X$ in $\R^d$ holds a close relationship with some region in $[0,1]^d$. Formally, if we define
\begin{equation} \label{eq:G_U}
	\mathcal{G}_{\textbf{U}} : = \left\{ \textbf{u} \in [0,1]^d \, : \, \big(F_1^{-1}(u_1),..., F_d^{-1}(u_d) \big) \in \mathcal{G}_\X \right\}
\end{equation}
then it holds that
$$\x \in \mathcal{G}_{\X} \Longleftrightarrow \textbf{u} \in \mathcal{G}_{\textbf{U}}.$$

Therefore, similarly to the simulation of an unconditional multivariate distribution, to sample $\x = (x_1,...,x_d)$ from the distribution of $\Big(\X \, | \, g(\X) \in A\Big)$ we can
\begin{enumerate}
	\item Produce a weighted sample $\{\textbf{u}^{(j)}, \, W^{(j)} \}_{j=1}^N$ from $C$ such that $\textbf{u}^{(j)} = (u_1^{(j)},...,u_d^{(j)}) \in \mathcal{G}_{\textbf{U}}$ for all $j=1,...,N$;
	\item Return the weighted sample $\{\x^{(j)}, \, W^{(j)} \}_{j=1}^N$ where $x_i^{(j)} = F_i^{-1}(u_i^{(j)})$, for $i=1,...,d$, $j=1,...,N$. %$x_1^{(j)} = F_1^{-1}(u_1^{(j)}), ..., x_d^{(j)} = F_d^{-1}(u_d^{(j)})$.
\end{enumerate}

Note that one can calculate conditional expectations with respect to $\X$ as follows 
\begin{align} \label{eq:conditionalExpectation_u}
\Exp[h(\X) \, | \, g(\X) \in A]  &= \Exp\Big[h\big(F_1^{(-1)}(u_1),...,F_d^{(-d)}(u_d) \, \big)| \, \textbf{u} \in  \mathcal{G}_{\textbf{U}}  \Big]  \\
                                 &\approx \sum_{j=1}^N W^{(j)} h\big(F_1^{(-1)}(u_1^{(j)}),...,F_d^{(-d)}(u_d^{(j)}) \big).
\end{align}

Clearly, if all the marginal quantile functions $F_i^{-1}$ are known, then the difficulty of the proposed approach is to sample from the constrained copula. The idea of performing the sampling procedure in the constrained copula space has been independently developed by Arbenz, Cambou and Hofert in \cite{arbenz2014importance} (see an overview of their algorithm in section \ref{sec:ACH}), where an importance sampling distribution is designed to target the distribution of $\textbf{u} \, | \, \textbf{u} \in \mathcal{G}_{\textbf{U}}$.

In this paper, instead of targeting the rare region $\textbf{u} \in \mathcal{G}_{\textbf{U}}$ we propose to sequentially target less rare regions, in a specially designed SMC Sampler procedure that is made precise in Sections \ref{sec:intermediateSets} and \ref{sec:SMC}.

Before presenting these specialized algorithms, it is informative to briefly comment on alternative Importance Sampling (IS) based approaches in the actuarial literature. At this stage we observe that there are many different types of rare-event simulation algorithms available, and the choice of a particular type will depend principally on how one defines the notion of a ``rare-event'' in the sample space. Although the following brief discussions on alternative IS based solutions are not directly targeting the same type of multivariate rare-event problems as faced in the case of capital allocation, they are informative to discuss especially with regard to the concept of relative error.

We also note that there are classes of asymptotic approximation results available for approximation of capital allocations. For instance, in order to estimate expectations of the form (\ref{eq:conditionalExpectation}) in a bivariate set-up, \cite{cai2014estimation} assume that large values of $g(\X)$ correspond to high values of $h(\X)$ (in their case $h(\X) = X_i$). Under these constraints, the authors use results from Extreme Value Theory (EVT) to derive an estimator of (\ref{eq:conditionalExpectation}) and study some of its properties. We refer the interested reader to references in \cite{cai2014estimation} for further background on such asymptotic approximations and instead we continue to focus upon sampling based solutions.

%================================================================================================
\subsection{\textbf{Related Monte Carlo and Importance Sampling approaches}}
%================================================================================================

%Several IS based solutions have been proposed in the literature for rare-event estimation in constrained settings, not necessarily directly in the settings described above. It is useful to recall a few of the basic alternatives here that have appeared in the actuarial and risk literature.

In the particular case of Gaussian copulas \cite{glasserman2005measuring} presents an IS scheme to approximate conditional expectations. For the same family of models, \cite{siller2013measuring} more recently proposed a method based on Fourier transforms to compute marginal risk contributions. In \cite{mcleish2010bounded} the author develops a class of IS based estimators that satisfy a condition of bounded relative error. In particular they estimate tail events in a univariate framework via IS based distributions constructed from exponential families which will guarantee bounded relative error in estimated tail functionals. The class of methods they develop revolves around exponential tilting of the tails of the target distribution, also known in the actuarial literature as the Esscher transform or tempering. In addition, it is argued in \cite{mcleish2010bounded} that for tail events, the variance or standard error is a less desirable quantity to consider in assessing the performance of algorithm, compared to a version scaled by the mean such as the relative error. Therefore, \cite{mcleish2010bounded} proposed consideration of the relative error which is simply the ratio of the estimators standard error deviation to its mean. In Section \ref{sec:variance} we present some discussion of this concept in the context of the SMC algorithms proposed.

The approach proposed in \cite{mcleish2010bounded} selects the IS distribution to minimize the relative error of the rare-event probability by rewriting the problem as the solution to the minimization of the Renyi generalized divergence. In several simple univariate examples one can prove the relative error can be bounded if selected in this manner. This is interesting as it is contrary to other IS based approaches which seek to minimize the importance sampling weights variance. The closest class of IS based solution to our proposed SMC Sampler solution has been developed recently by \cite{arbenz2014importance}. We present this briefly before detailing our approach.

%================================================================================================
\subsection{\textbf{Arbenz-Cambou-Hofert algorithm}} \label{sec:ACH}
%================================================================================================
In common with the proposed sampling method in this paper, the approach recently developed in \cite{arbenz2014importance} involves sampling from a target distribution given by a constrained copula. This is particularly relevant as it leads to convenient bounded state spaces for sampling, since the support of the copula when unconstrained is $[0,1]^d$ and consequently the constrained copula will be on a sub-space of this hypercube.

Unlike our proposed solution, the approach of \cite{arbenz2014importance} involves developing an Importance Sampling (IS) scheme targeting the constrained copula distribution. However, in contrast to our approach their method does not involve any intermediate sequence of constrained regions leading smoothly up to the rare-event constraint. Instead, they try to directly approximate the optimal importance sampling proposal. This is in general a very challenging task and they have some interesting insight. For this reason we briefly present their methodology below as it will form a direct comparison with the approach we develop.

In \cite{arbenz2014importance} the aim of their work is to generate a sample from the unconditional copula with most of the particles satisfying a condition such as (\ref{eq:G_U}). In order to generate these samples, an importance sampling distribution $F_\textbf{V}$ is designed as a mixture of conditional copulas. More formally, the IS distribution is defined as
\begin{equation}
F_\textbf{V}(\textbf{u}) = \int_0^1 C^{[\lambda]}(\textbf{u}) dF_\Lambda(\lambda),
\label{eq:f_v_arbenz}
\end{equation}
where $C^{[\lambda]}$ is the distribution of $\textbf{U}$ conditional on the event that at least one of its components exceeds $\lambda$, ie,
$$C^{[\lambda]}(\textbf{u}) = \Prob[U_1\leq u_1,...,U_d\leq u_d \, | \, \max\{U_1,...,U_d \} > \lambda].$$

In the main algorithm presented in \cite{arbenz2014importance}, samples from the importance distribution are generated by rejection, but a ``conditional sampling algorithm'' is also presented. Overall, an appealing aspect of their proposed method is that it does not make use of the copula density explicitly. This can be advantageous in settings in which the copula density is computationally expensive to be calculated or even unknown. However, as with all Monte Carlo methods, there are also drawbacks to the proposed approach that we will argue can be overcome through development not of an IS solution but instead via a Sequential Monte Carlo Sampler (SMC Sampler) solution in the constrained copula space.

Under simplifying assumptions on the joint behaviour of $\textbf{U}$, an optimal distribution for $F_\Lambda$ is presented, but in general the only restriction on the choice of the mixing distribution $F_\Lambda$ is that $\Prob[\Lambda = 0] > 0$. To sample from $\X \, | \, \X > B$ one of the algorithms proposed is given as follows.

%-----------------------------------------------------------------------------------------------
{\small
\begin{algorithm}[H]
 \SetAlgoLined
 \textbf{Inputs:} N: desired sample size for $F_\textbf{V}$; $F_\Lambda$: mixing distribution (as in (\ref{eq:f_v_arbenz})) \;
 	\For{$j=1,...N$}{
		Sample $\Lambda^{(j)} \sim F_\Lambda$ \;
		Sample $\textbf{u}^{(j)} \sim C$ until $\max\{u_1^{(j)},...,u_d^{(j)} \} > \Lambda^{(j)}$ \;
		Define $x_i^{(j)}:= F_i^{-1} \left( u_i^{(j)} \right)$ for $i=1,...,d$ \;
		Compute the importance weight $w^{(j)}: = w(\textbf{u}^{(j)})$ as in \cite{arbenz2014importance}, Section 5 \;
		Compute the normalized importance weight $W^{(j)} = \frac{w^{(j)}}{\sum_{j=1}^N w^{(j)}}$ \;
		Define $\{ \widetilde{\textbf{u}}^{(j_k)} \}_{k=1}^{\widetilde{N}}$ as the sub-set of $\{ \textbf{u}^{(j)} \}_{j=1}^N$ such that  $\sum_{i=1}^d \widetilde{x}_i^{(j_k)} > B$\;
	}
	%Return $\hat{\mu}_i^{IS} = \sum_{k=1}^{\widetilde{N}} u_i^{(j_k)} W_i^{(j_k)}$ for $i=1,...,d$.
	\KwResult{Weighted random samples (of random sample size $\widetilde{N}$): $\big\{ \textbf{u}^{(j)}, \ W^{(j)} \big\}_{j=1}^{\widetilde{N}}$;}
\caption{IS-ACH algorithm from \cite{arbenz2014importance}.}
\label{algo:IS_arbenz}
\end{algorithm}
}
	%\vspace{5mm}
%-----------------------------------------------------------------------------------------------

Some points need to be stressed about Algorithm \ref{algo:IS_arbenz}. First, let $\Exp[N_\textbf{V}]$ denote the expected number of draws from $C$ in order to have a sample satisfying $\max\{u_1,...,u_d \} > \Lambda$. It can be easily shown that (see \cite{arbenz2014importance}, Lemma 4.2)
$$\Exp[N_\textbf{V}] = \int_0^1  \frac{1}{1- C(\lambda \textbf{1})} dF_\Lambda(\lambda),$$
where $\textbf{1} = (1,...,1) \in \R^d$. Therefore, to generate a sample of fixed size $N$ from $F_\textbf{V}$ it is necessary to sample (on average) $N \times \Exp[N_\textbf{V}]$ times from $C$.

Another important aspect of this algorithm pertains to an understanding of the number of ``particles'' (samples) which are obtained with non-zero weight. Since $p_0 := \Prob[\Lambda=0]$ is necessarily positive, we can ensure some of the $N$ samples from $F_\textbf{V}$ will be actually from the unconditional copula $C$, meaning that $p_0 \times 0.99 \times 100\%$ of the particles are expected not to satisfy the condition $\sum_{i=1}^d X_i > VaR_{0.99}(\sum_{i=1}^d X_i)$. For $\lambda > 0$ the same behaviour is expected, leading, in practice, to $\widetilde{N}$ (as defined in the last step of Algorithm \ref{algo:IS_arbenz}) being smaller than $N$, and in cases of relevance to capital allocation, this difference can be significant, with $\widetilde{N} << N$. In capital allocation problems such cases can prove to be a serious problem in terms of computational cost and efficiency for this IS based approach as will be discussed in Section \ref{sec:Examples}.

%\textcolor{red}{Rodrigo: check if Hoffert et al say anything about their algorithm and relative error being bounded......}

Next we will present a completely different class of methods to the IS based solutions discussed. These will be based on a class of algorithms that extends IS solutions to sequential settings, known in statistics literature as Sequential Monte Carlo Samplers (SMC Samplers). To understand SMC Samplers we first recall the SMC algorithm before showing its generalization to the class of SMC Sampler algorithms.

%================================================================================================
%================================================================================================
\section{Reaching rare-events through sequences of intermediate sets} \label{sec:intermediateSets}
%================================================================================================
%================================================================================================

The idea of using intermediate sets to approximate the conditional density $f_{\X | g(\X) \in A}(\x)$ is to start sampling from the unconditional distribution $f_\X(\x)$ and move the weighted particles towards the rare conditioning set through ``not so rare'' sets.

Using the notation from the previous section, for a fixed function $g$ and set $A$, let $\{A_t\}_{t=1}^T$ be a sequence of nested sets shrinking to $A$, ie, $A_t \subset A_{t-1}$ and $A_t \downarrow A$, when $t \rightarrow T$. This sequence of sets defines a sequence of regions (as before):
\begin{align*}
	\mathcal{G}_{\X_t} &: = \{ \x_t \in \R^d \, : \, g(\x_t) \in A_t \}, \\
	\mathcal{G}_{\textbf{U}_t} &: = \left\{ \textbf{u}_t \in [0,1]^d \, : \, \big(F_1^{-1}(u_{t,1}),..., F_d^{-1}(u_{t,d}) \big) \in \mathcal{G}_{\X_t} \right\}.
\end{align*}

Although it is true that
$$\x_t \in \mathcal{G}_{\X_t} \Longleftrightarrow \textbf{u}_t \in \mathcal{G}_{\textbf{U}_t},$$
with $\textbf{u}_t = (u_{t,1},...,u_{t,d}):= \big(F_1(x_{t,1}),...,F_d(x_{t,d}) \big)$ we will see in the sequel that working in the bounded space $[0,1]^d$ we will have some advantages in the design of the algorithm. In this set-up our goal will be ultimately to have (weighted) samples $\left\{\textbf{u}^{(j)}_T, \, W_T^{(j)} \right\}_{j=1}^N$ from the conditional copula $c(\textbf{u}_T \, | \textbf{u}_T \in \mathcal{G}_{\textbf{U}_T})$ which would be then transformed through the marginal inverse cdf's in order to get a weighted sample from $f_\X(\x \, | \, \x \in \mathcal{G}_{\X_T})$. 

Following the notation to be used in Section \ref{sec:SMCsamplers} we define our target distribution at each time step (level) $t=1,...,T$ as
\begin{equation}
\label{eq:target}
\pi_t(\textbf{u}_t):= \frac{c(\textbf{u}_t)\one_{\left\{\textbf{u}_t \in \mathcal{G}_{\textbf{U}_t} \right\}}  (\textbf{u}_t) }{\Prob[\textbf{U}_t \in \mathcal{G}_{\textbf{U}_t}]}.
\end{equation}

%\begin{figure}%
%\includegraphics[width=1\textwidth]{./Figures/copula_joint_constraints_PAINT.pdf}
%\caption{}%
%\label{}%
%\end{figure}

%================================================================================================
\subsection{\textbf{Copula Constrained Geometry}} \label{sec:copSMCn}
%================================================================================================
Before we formalize the algorithm to sample from the constrained copula we will study some properties of the restricted region in the copula space, defined in (\ref{eq:G_Z}) and (\ref{eq:G_U}), for the particular case where $g(\X) = \sum_{i=1}^d X_i$ and $A =[B, +\infty)$. The idea is that the knowledge of the restricted region can help us to design more efficient sampling schemes. Similar analyses can be performed for different restrictions.

In $\R^d$ our interest is to study points such that $\sum_{i=1}^d x_i = B$ which turn out to be equivalent to points in $[0,1]^d$ such that $\sum_{i=1}^d F_i^{-1}(u_i) = B$. It is easy to see that each of these curves (in $\R^d$ or in $[0,1]^d$) will lie in a $d-1$ dimensional space. Formally, these curves are defined through the following mappings
\begin{equation*} %\label{eq:}
	\widetilde{\mathcal{G}}_\X : = \left\{ (x_1,...,x_{d-1}) \in \R^{d-1} \, : \, (x_1,...,x_{d-1}, B- \sum_{i=1}^{d-1}x_i)  \right\},
\end{equation*}

\begin{equation} \label{eq:Gtilde_U}
	\widetilde{\mathcal{G}}_\textbf{U} : = \left\{ \textbf{u}_{-d} := (u_1,...,u_{d-1}) \in [0,1]^{d-1} \, : \, \big(u_1,...,u_{d-1}, r(\textbf{u}_{-d}) \big)  \right\},
\end{equation}
where $r(\textbf{u}_{-d}): = F_d\Big( B - \sum_{i=1}^{d-1} F_i^{-1} (u_i) \Big)$.

First of all, note that if $g$ is a generic continuous function  and all the marginal cdfs $F_1,...,F_d$ are continuous then the curve $\widetilde{\mathcal{G}}_\textbf{U}$ (defined similarly to (\ref{eq:Gtilde_U}) ) will be continuous. Moreover, the region $\mathcal{G}_\textbf{U}$ in (\ref{eq:G_U}) will not be the union of disjoint set, but only one continuous region. Some other properties of these regions may be derived in particular cases. For example, we know that in the linear case ($g(\X) = \sum_{i=1}^d X_i$) the curve in $[0,1]^d$ will pass through the points $\big(F_1(B), 0, ..., 0\big), \, \big(0, F_2(B), 0,...,0\big), ..., \, \big(0,...,0,F_d(B)\big)$.

Another interesting information about the curve $\widetilde{\mathcal{G}}_\textbf{U}$ is given by its curvature, as seen in the next Proposition.

\begin{proposition} \label{prop:generalConvexity} The curve $\widetilde{\mathcal{G}}_{\textbf{U}}$ defined in $(\ref{eq:Gtilde_U})$ is convex at $(u_1,..,u_{d-1}, r(\textbf{u}_{-d}))$ if
%$r(\textbf{v}_{-d}) - r(\textbf{u}_{-d}) - \langle \nabla r(\textbf{u}_{-d}), \, \textbf{v}_{-d} - \textbf{u}_{-d}^1 \rangle \geq 0, \ \  \forall \textbf{u}_{-d}, \textbf{v}_{-1} \in \R^{d-1}$;
	$$\langle \textbf{u}_{-d}, \, \nabla^2 r(\textbf{u}_{-d}) \textbf{u}_{-d} \rangle >  0, \ \ \forall \textbf{u}_{-d} \in \R^{d-1},$$
where $\langle \textbf{x}, \textbf{y} \rangle$ is the inner product of $\textbf{x}$ and $\textbf{y}$ and $\nabla^2 f$ is the Hessian matrix of $f$.
\end{proposition}

In the particular case where $r(\textbf{u}_{-d}): = F_d\Big( B - \sum_{i=1}^{d-1} F_i^{-1} (u_i) \Big)$ the general terms of the Hessian matrix are given by
\begin{align*}
&\frac{\partial^2 r}{ \partial u_j\partial u_k} (\textbf{u}_{-d}) =  \frac{f_d'\left(B- \sum_{i=1}^{d-1} x_i \right) }{f_j(x_j) f_k(x_k)}, \  \ \forall j \neq k, \ \ j,k=1,...,d-1 \\
&\frac{\partial^2 r}{ \partial u_j^2} (\textbf{u}_{-d}) = \frac{ f_d'\left(B- \sum_{i=1}^{d-1} x_i \right)  f_j(x_j)  + f_d\left(B- \sum_{i=1}^{d-1} x_i \right)  f_j'(x_j)}{  \left[ f_j(x_j)  \right]^3  }, \\ 
& \quad \forall j=k, \ \ j=1,...,d-1
\end{align*}
where, once again, we use the notation $x_i = F_i^{-1}(u_i)$ to make the above formulas more appealing.

From Proposition \ref{prop:generalConvexity}, in the very particular case where $d=2$ and $X_1,X_2 \geq 0$ (representing losses, for example) the concavity of $\widetilde{\mathcal{G}}_{\textbf{U}}$ is determined only by the sign of 
$$ f_2'\left(B- x_1 \right)  f_1(x_1)  + f_2\left(B- x_1 \right)  f_1'(x_1).$$
This is due to the fact that the denominator is the power of a density function (non-negative) and that $x_1$ is non-negative.

On Figure \ref{fig:constraitRegions} we can see that for different constraint levels the curve in $[0,1]^2$ presents different shapes, continuously varying from a convex to a concave region.

\begin{figure}%
\includegraphics[width=1\textwidth]{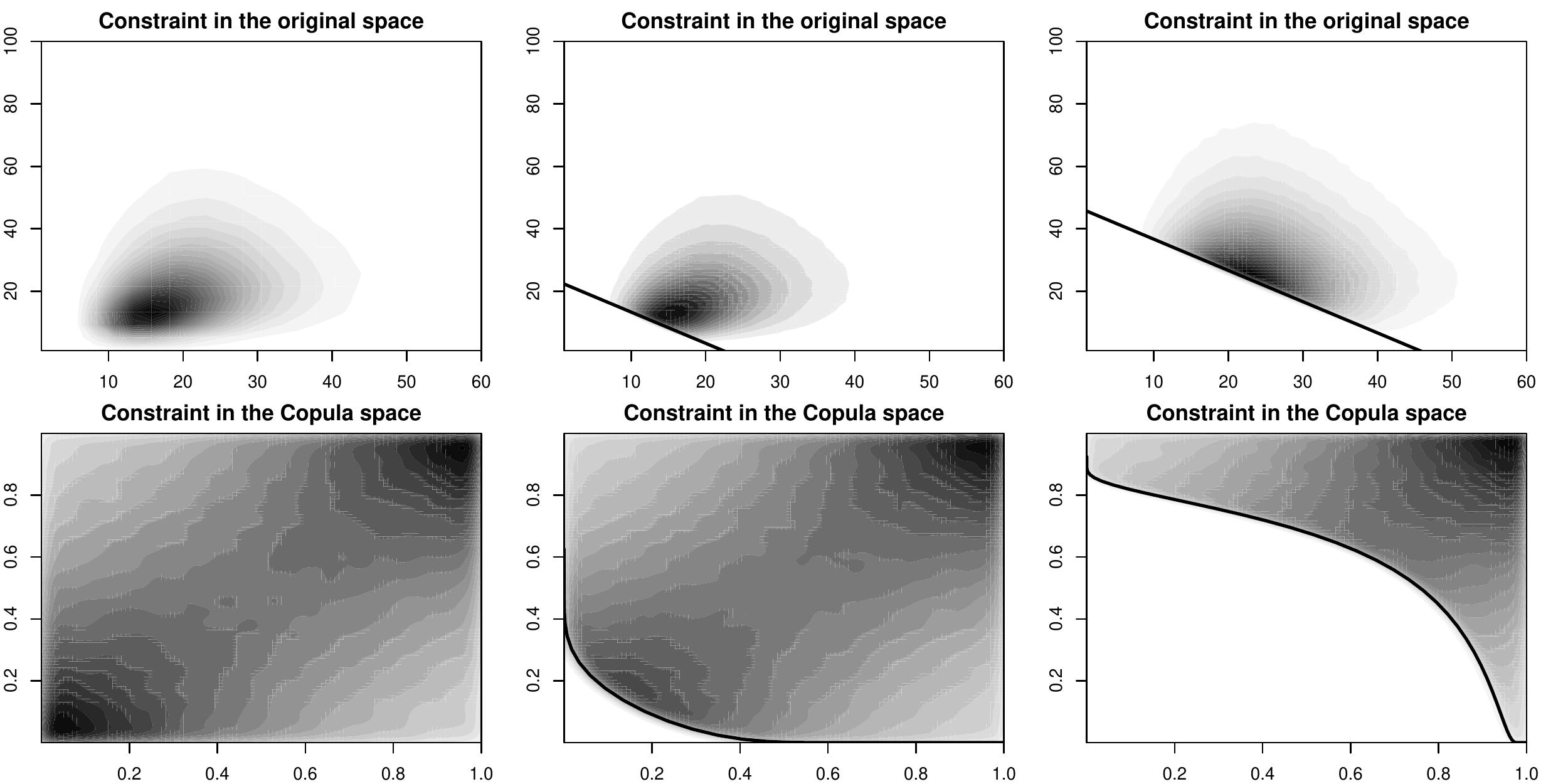}
\caption{\small Frank copula with parameter 2, Log-Normal marginal cdf's, both with same $\mu=3$, and $\sigma_1 = 0.4$, $\sigma_2 = 0.6$. (Top row) Constraint in the original space, for $B=0, 23, 46$ (Bottom row) Constraint in the Copula space, $[0,1]^2$, for equivalent levels.}
\label{fig:constraitRegions}
\end{figure}

%\FloatBarrier
%================================================================================================
%================================================================================================
\section{Sequential Monte Carlo Methods (SMC) \label{SMC}}
\label{sec:SMC}
%================================================================================================
%================================================================================================
In this section we will introduce the general class of algorithms known as Sequential Monte Carlo (SMC) and an important variant for rare-event simulation, the SMC Samplers classes of methods. This family of Monte Carlo algorithms has been developed to approximate sequences of integrals constructed from a sequence of probability density functions. Of course, adjustments are possible when the interest lies only in one distribution, such as the terminal distribution in a sequence of intermediate increasingly rare events, such as was shown in the ideas presented in Section \ref{sec:intermediateSets}).

SMC methods have  emerged out of the fields of engineering, probability and statistics in recent years. Variants of the methods sometimes appear under the names of particle filtering or interacting  particle systems e.g. \cite{ristic2004beyond}, \cite{doucet2001sequential}, \cite{del2004feynman}, and their theoretical properties have been extensively studied  in \cite{crisan2002survey}, \cite{del2004feynman}, \cite{Chopin2004}, \cite{kunsch2005recursive}. For a recent survey in the topic, with focus on economics, finance and insurance applications the reader is referred to \cite{creal2012survey} and \cite{del2013introduction}. For an application to rare events in a financial context, see \cite{carmona2009interacting}.

The general context of a standard SMC method is that one wants to approximate a (often naturally occurring) sequence of probability density functions (pdf's) $\big\{ \widetilde{\pi}_t \big\}_{t \geq 1}$ such that the support of every function in this sequence is defined as $supp\big( \widetilde{\pi}_t \big) = E_t$ and the dimension of $E_t$ forms an increasing sequence, i.e., $dim \big(E_{t-1} \big) < dim \big( E_t \big)$. For example, the reader can think of $E_1 = \R^d, ..., E_t = \R^{d\times t}$, which will be precisely the sequence to be used throughout this work.

We may also assume that $\widetilde{\pi}_t$ is only known up to a normalizing constant, 
$$\widetilde{\pi}_t(\x_{1:t}) = Z_t^{-1} \widetilde{f}_t(\x_{1:t}),$$
where $\x_{1:t} := (\x_1,...,\x_t) \in E_t = \R^{d\times t}$. As in section \ref{sec:simForConditionalExp}, the approximation for $\widetilde{\pi}_t$ is given by a weighted sum of random samples (also known as ``particles'').

Procedurally, we initialize the algorithm sampling a set of $N$ particles from the distribution $\widetilde{\pi}_1$ and set the normalized weights to $W_1^{(j)}= 1/N$, for all $j=1,...,N$. If it is not possible to sample directly from $\widetilde{\pi}_1$, one should sample from an importance distribution $\widetilde{q}_1$ and calculate its weights accordingly (see Algorithm \ref{algo:SMC}). Then the particles are sequentially propagated thorough each distribution $\widetilde{\pi}_t$ in the sequence via three main processes: mutation, correction (incremental importance weighting) and resampling. In the first step (mutation) we propagate particles from time $t-1$ to time $t$ and in the second one (correction) we calculate the new importance weights of the particles.

This method can be seen as a sequence of IS steps, where the target distribution at each step $t$ is $\widetilde{f}_t(\x_{1:t})$ (the unnormalized version of $\widetilde{\pi}_t$) and the importance distribution is given by 
\begin{equation}
\widetilde{q}_t(\x_{1:t}) =\widetilde{q}_1(\x_1)\prod_{j=2}^t K_j(\x_{j-1},\x_j),
\label{eq:qTilde}
\end{equation}
where $K_j(\x_{j-1}, \, \cdot \, )$ is the mechanism used to propagate particles from time $t-1$ to $t$, known as the mutation stage. The algorithm works in the following way:
\vspace{0.2cm}

%---------------------------------------------------------------------------------------------------------------------------
{\small
\begin{algorithm}[H]
 \SetAlgoLined
 \textbf{Inputs:} IS density $\widetilde{q}_1$, (forward) mutation kernels $\big\{K_t(\x_{t-1}, \x_t) \big\}_{t=1}^T$\;
 	\For{$j=1,...,N$}{
		Sample $\x_1^{(j)}$ from $\X_1 \sim \widetilde{q}_1(\, \cdot \,)$ \emph{(Mutation step)}\;
		Calculate the weights $w_1^{(j)} = \frac{\widetilde{f}_1(\x_1^{(j)})}{\widetilde{q}_1(\x_1^{(j)})}$\;
	}
	Calculate the normalized weights $W_1^{(j)} = \frac{w_1^{(j)}}{\sum_{j=1}^N w_1^{(j)}}$ \emph{(Correction step)}\;
	\For{$t=2,\hdots,T$}{
		\For{$j=1,...,N$}{
		Sample $\x_t^{(j)}$ from $\X_t \big| \X_{t-1} = \x_{t-1}^{(j)} \sim K_t(\x_{t-1}^{(j)}, \ \cdot \ )$ \emph{(Mutation step)}\;
		Create the vector $\x_{1:t}^{(j)}:= (\x_{1:(t-1)}^{(j)}, \, \x_t^{(j)})$ \;
		Calculate the weights $\displaystyle w_t^{(j)} = \frac{\widetilde{f}_t(\x_{1:t}^{(j)})}{\widetilde{q}_t(\x_{1:t}^{(j)})} = w_{t-1}^{(j)} \underbrace{  \frac{\widetilde{f}_t(\x_{1:t}^{(j)})}{\widetilde{f}_{t-1}(\x_{1:t-1}^{(j)})K_t(\x_{t-1}^{(j)},\x_t^{(j)})}  }_{\text{incremental weight: } \widetilde{\alpha}(x_{1:t}^{(j)})}$ \;
	 }
	Calculate the normalized weights $W_t^{(j)} = \frac{w_t^{(j)}}{\sum_{j=1}^N w_t^{(j)}}$ \emph{(Correction step)}.
	}
	\KwResult{Weighted random samples $\big\{ \x_{1:t}^{(j)}, \ W_t^{(j)} \big\}_{j=1}^N$ approximating $\widetilde{\pi}_t$, for all $t=1,...,T$;}
\label{algo:SMC}
\caption{Standard SMC algorithm.}
\end{algorithm}
}
%---------------------------------------------------------------------------------------------------------------------------

If $\big\{ \x_{1:t}^{(j)}, \ W_t^{(j)} \big\}_{j=1}^N$ is a set of weighted particles returned by the SMC algorithm then  
\begin{equation}
\sum_{j=1}^N W_t^{(j)} \varphi(\x_{1:t}^{(j)}) \longrightarrow \Exp_{\widetilde{\pi}_t}[\varphi(\X_{1:t}) ] := \int_{E^t} \varphi(\x_{1:t}) \widetilde{\pi}_t(\x_{1:t}) d\x_{1:t},
\label{eq:SMC_approximation}
\end{equation}
$\widetilde{\pi}_t$--almost surely as $N \rightarrow +\infty$, for any test function $\varphi$ such that the expectation of $\varphi$ under $\widetilde{\pi}_t$ exists.

\begin{remark} The reader should note that the knowledge of $\widetilde{\pi}_t$ up to a normalizing constant is sufficient for the implementation of a generic SMC algorithm, since the normalized version of the weights would be the same for both $\widetilde{\pi}_t$ and $\widetilde{f}_t$.
\end{remark}

The optimal selection of the mutation kernel (SMC importance distribution) for SMC methods is widely studied and a good tutorial review on the optimal choice minimizing the variance of the incremental importance sampling weights is overviewed in \cite{doucet2009tutorial}. There are also a range of known probabilistic properties of the SMC algorithm available in the literature, for a tutorial in the insurance context on these properties see \cite{del2013introduction}. This includes details on central limit theorem results for SMC algorithms along with asymptotic variance expressions, finite sample bias decompositions and propagation of chaos as well as finite sample concentration inequality bounds. There are also tutorials available on SMC algorithms in general such as \cite{doucet2009tutorial} and the book length coverage of \cite{del2004feynman}.

%================================================================================================
\subsection{\textbf{Resampling and Moving particles}} \label{ResampleMove}
%================================================================================================
In practice the generic algorithm presented in the previous section will eventually (as $t$ increases) be based only in a few distinct particles, in the sense that almost all the other particles will have negligible weights. This phenomenon is known as particle degeneracy. 

To overcome this problem, when the system is too degenerate one can resample all the particles $x_{1:t}$ after the correction step, choosing the $j$-th one with probability proportional to $W_t^{(j)}$. In \cite{liu1998sequential} it was suggested using the Effective Sample Size (ESS) to measure the sample degeneracy, where
$$ESS_t : = \left[ \sum_{j=1}^N (W_t^{(j)})^2 \right]^{-1}$$
and resample steps should be performed only when $ESS_t < M $ -- as a rule of thumb we can set $M = N/2$. It is important to note that after this step we need to set $W_t^{(j)} = 1/N$ for all particles, since they are all identically distributed.

Although the resample step alleviates the degeneracy problem, its successive reapplication, at each stage of the sampler, produces the so-called sample impoverishment, where the number of distinct particles is extremely small. In \cite{gilks2001following} it was proposed to add a move with any kernel such that the target distribution is invariant with respect to it to rejuvenate the system. This kernel may be, for example, a Markov Chain kernel, which would begin with equally weighted samples from the target distribution and then perturb them under a single step of a Metropolis Hastings accept-reject mechanism. This would preserve the target distribution and add diversity to the particle cloud. 

More precisely, we can apply any kernel $M(\x_{1:t}, \, \x_{1:t}^*)$ that leaves $\widetilde{\pi}_t$ invariant to move particle $\x_{1:t}$ to $\x_{1:t}^*$ (the star will denote particles \emph{after} the ``move'' step), ie,
$$\widetilde{\pi}_t(\x_{1:t}^*) = \int M(\x_{1:t}, \, \x_{1:t}^*) \widetilde{\pi}_t(\x_{1:t}) d\x_{1:t}.$$

Two of the simplest ways to construct such a kernel $M$ are to use a Gibbs sampler or a Metropolis-Hastings (M-H) algorithm. To use a Gibbs sampler algorithm, the full conditional distributions \\
$\widetilde{\pi}_t(\x_{1:t,i} \, | \, \x_{1:t,1},...,\x_{1:t,i-1},\x_{1:t,i+1},...,\x_{1:t,d})$, for $i=1,...,d$, must be known up to proportionality, while for the M-H they are not necessary. On the other hand, in the M-H algorithm one needs to design a proposal density $Q(\x_{1:t}, \x_{1:t}^*)$ that moves the particle $\x_{1:t}$ to $\x_{1:t}^*$ or some component of it such as $\x_t$ to $\x_t^*$. The Gibbs sampler is presented as Algorithm \ref{algo:gibbs_sampler}. For the Metropolis-Hastings and/or more details on MCMC methods, see, for example,\cite{gamerman2006markov}. 

\vspace{0.2cm}
{\small
\begin{algorithm}[H]
 \SetAlgoLined
 %\KwData{Model parameters $(\lambda, \mu, \lambda_N)$; Number of simulations $(M)$}
 \textbf{Inputs:} Full conditional pdf's: $\widetilde{\pi}_t(\x_{1:t,1} \, | \, \x_{1:t,2},...,\x_{1:t,d}),...,\widetilde{\pi}_t(\x_{1:t,d} \, | \, \x_{1:t,1},...,\x_{1:t,d})$; Sample from $\widetilde{\pi}_t$: $\x_{1:t} =(\x_{1:t,1},...,\x_{1:t,d})$ \;
	Sample $\x_{1:t,1}^* \sim \widetilde{\pi}_t(\x_{1:t,1} \, | \, \x_{1:t,2},...,\x_{1:t,d})$ \;
	Sample $\x_{1:t,2}^* \sim \widetilde{\pi}_t(\x_{1:t,2} \, | \, \x_{1:t,1}^*,\x_{1:t,3},...,\x_{1:t,d})$ \;
	\vdots
	Sample $\x_{1:t,d}^* \sim \widetilde{\pi}_t(\x_{1:t,d} \, | \, \x_{1:t,1}^*,...,\x_{1:t,d-1}^*)$ \;
	\KwResult{New sample from $\widetilde{\pi}_t$: $\x_{1:t}^* =(\x_{1:t,1}^*,...,\x_{1:t,d}^*)$ ;}
\label{algo:gibbs_sampler}
\caption{Gibbs Sampler algorithm.}
\end{algorithm}
}

\vspace{0.2cm}
Including both the resampling and the ``move'' steps into the generic SMC algorithm leads to the ``Resample-Move'' algorithm, first presented in \cite{gilks2001following} and subsequently widely used in the SMC literature.
%{\color{red} Different resampling schemes and thresholds (Adam Johansen's work), In \cite{murray2013rethinking} new resampling schemes that are suitable for parallel implementation are presented.}

The generic class of SMC algorithms whilst widely used in practice can not be directly applied to the problems addressed in this work, since all the distributions in the sequence (\ref{eq:target}) are defined over the same support, ie, $E_t = E$ and not $E_t = E \times ... \times E$ as required by the SMC algorithms just described. To overcome this problem a specialized variation of this method, named SMC Samplers is introduced in the next section.
%\FloatBarrier

%================================================================================================
\subsection{\textbf{SMC Samplers}} \label{sec:SMCsamplers}
%================================================================================================
Having presented the SMC class of algorithms, we now present in contrast to these the class of SMC Sampler algorithms which involve the same mechanism as the SMC algorithm also using a mutation, correction and resampling stage at each iteration. However, the class of SMC Sampler algorithm is importantly different in the space that the sequence of distributions being sampled from are defined upon. Differently from Section \ref{SMC}, our interest now is to approximate a generic sequence of probability distributions $\{\pi_t\}_{t=1}^T$ such that $supp(\pi_t) = supp(\pi_{t-1}) = E$ for all $t=1,...,T$ (once again one may think of $E = \R^d$). Here we may also assume our target distribution is only known up to a normalizing constant, ie, $\pi_t(\x_t) = Z_t^{-1} f_t(\x_t).$ For notational clarity, functions in the enlarged space will be denoted with a upper tilde, as $\widetilde{f}_t : E^t \longrightarrow \R$.

The idea presented on \cite{peters2005topics} and \cite{DelMoral2006} is to transform this problem into one resembling the usual SMC formulation, where the sequence of target distributions $\{ \widetilde{f}_t \}_{t =1}^T$ is defined on the product space, i.e., $supp(\widetilde{f}_t) = E \times E \times ... \times E = E^t$. 

The construction of $\widetilde{f}_t$ (the density in the \emph{path} space) is carried out as:
\begin{equation} \label{eqn:enlarged_f}
 \widetilde{f}_t(\x_{1:t}) = f_t(\x_t)\widetilde{f}_t(\x_{1:t-1}| \x_t), \text{ for } t=2,...,T
\end{equation}
where $\widetilde{f}_t(\x_{1:t-1}| \x_t)$ is a probability distribution on the space $E^{t-1}$, for all $\x_t \in E$. Similarly to (\ref{eq:qTilde}) we can carry out the construction of the importance distribution at time $t$. 

As noticed in \cite{peters2005topics} and \cite{DelMoral2006}, a wise choice for $\widetilde{f}_t(\x_{1:t-1}| \x_t)$ is given by
$$\widetilde{f}_t(\x_{1:t-1}| \x_t) = \prod_{s=1}^{t-1}L_s(\x_{s+1}, \x_s),$$
where the each $L_s$ is the density of an (artificial) backward Markov kernel. It is important to note that, by construction, $\widetilde{f}_t(\x_{1:t})$ admits $f_t(\x_t)$ as a marginal, since
$$\int \widetilde{f}_t(\x_{1:t}) d\x_{1:t-1} = f_t(\x_t) \int \prod_{s=1}^{t-1}L_s(\x_{s+1}, \x_s) d\x_{1:t-1} = f_t(\x_t), \ \forall t>1.$$

Moreover, provided that $\widetilde{f}_t$ admits $f_t$ as a marginal the normalizing constant of the enlarged density will be the same as the original density:
$$\int \widetilde{f}_t(\x_{1:t}) d\x_{1:t} = \int \int \widetilde{f}_t(\x_{1:t}) d\x_{1:t-1} d\x_t = \int f_t(\x_t) d\x_t = Z_t.$$

Now that we are back to the SMC framework from last section, we can easily write the SMC Sampler algorithm (Algorithm \ref{algo:SMC_sampler}). Moreover, the Resample-Move strategy from Section \ref{ResampleMove} can still be utilized.

%---------------------------------------------------------------------------------------------------------------------------
{\small
\begin{algorithm}%[H]
 \SetAlgoLined
 \textbf{Inputs:} IS density $q_1$, (forward) mutation kernels $\big\{K_t(\x_{t-1}, \x_t)\big\}_{t=1}^T$, (artificial) backward kernels $\big\{L_{t-1}(\x_{t}, \x_{t-1})\big\}_{t=1}^T$, move kernel $M_t(\widehat{\x}_t, \, \x_t)$\;
 	\For{$j=1,...,N$}{
		Sample $\x_1^{(j)}$ from $\X_1 \sim \widetilde{q}_1(\, \cdot \,)$ \emph{(Mutation step)}\;
		Calculate the weights $w_1^{(j)} = \frac{\widetilde{f}_1(\x_1^{(j)})}{\widetilde{q}_1(\x_1^{(j)})}$\;
	}
	Calculate the normalized weights $W_1^{(j)} = \frac{w_1^{(j)}}{\sum_{j=1}^N w_1^{(j)}}$ \emph{(Correction step)}\;
	\For{$t=2,\hdots,T$}{
		\For{$j=1,...,N$}{
		Sample $\x_t^{(j)}$ from $\X_t \big| \X_{t-1} = \x_{t-1}^{(j)} \sim K_t(\x_{t-1}^{(j)}, \ \cdot \ )$ \emph{(Mutation step)}\;
	Using (\ref{eqn:enlarged_f}) and (\ref{eq:qTilde}), calculate the weights $\displaystyle w_t^{(j)} = \frac{\widetilde{f}_t(\x_{1:t}^{(j)})}{\widetilde{q}_t(\x_{1:t}^{(j)})} = w_{t-1}^{(j)} \underbrace{\frac{f_t(\x_t^{(j)}) {L_{t-1}(\x_t^{(j)}, \x_{t-1}^{(j)}) } }{f_{t-1}(\x_{t-1}^{(j)})K_t(\x_{t-1}^{(j)},\x_t^{(j)})}}_{\text{incremental weight: } \alpha(\x_{t-1}^{(j)}, \x_t^{(j)})}$ \;
	 }
	Calculate the normalized weights $W_t^{(j)} = \frac{w_t^{(j)}}{\sum_{j=1}^N w_t^{(j)}}$ \emph{(Correction step)}. \;
  \If{$ESS_t < N/2$}{
		\For{j=1,...,N}{
			Resample $\widehat{\x}_t^{(j)} = \x_t^{(k)}$ with prob. $W_t^{(k)}$ \emph{(Resample step)} \;
			Sample $\x_t^{(j)} \sim M_t(\widehat{\x}_t^{(j)}, \, \cdot \, )$ \emph{(Move step)}\;
			Set $W_t^{(j)} = 1/N$ \;
		}
	}
 }
	\KwResult{Weighted random samples $\big\{ \x_{t}^{(j)}, \ W_t^{(j)} \big\}_{j=1}^N$ approximating $\pi_t$, for all $t=1,...,T$;}
\label{algo:SMC_sampler}
\caption{SMC Sampler algorithm.}
\end{algorithm}
}
%---------------------------------------------------------------------------------------------------------------------------

%================================================================================================
\subsubsection{Backward kernels selection} \label{optimalBackwardKernels}
%================================================================================================
The introduction of the sequence of kernels $\{L_{t-1}\}_{t=2}^T$ creates a new degree of freedom in SMC samplers when compared with usual SMC algorithms, where only the forward mutation kernels $\{K_t\}_{t=1}^T$ should be designed. In this section we will discuss how to, given the kernels $\{K_t\}_{t=1}^T$, optimize the choice of backward kernels $\{L_{t-1}\}_{t=2}^T$.

Denote by $q_t(\x_t)$ the marginal importance distribution at time $t$, which is given by
\begin{equation}
q_t(\x_t) = \int \widetilde{q}_t(\x_{1:t}) d\x_{1:t-1} = \int q_1(\x_1) \prod_{j=2}^t K_j(\x_{j-1}, \x_j) d\x_{1:t-1}.
\label{eq:q_t}
\end{equation}

In the case in which we know how to calculate $q_t$ in exact form we can simply approximate the target distribution $f_t$ by a weighted sample $\{\x_t^{(j)}, W_t^{(j)}\}$, where $\x_t \sim q_t$ and $W_t$ is the normalized version of $w_t := \frac{f_t(\x_t)}{q_t(\x_t)}.$ From the definition of $q_t$ we can see that sampling from $q_t$ is simple if it is easy to sample from $q_1$ and from all the kernels $K_t$. On the other hand, the density of $q_t$ will only be tractable if we are able to solve the marginalization integral (in $t-1$ dimensions) -- which, in practice will hardly ever be the case.

The introduction of backward kernels $\{L_{t-1}\}_{t=2}^T$ helps us (in most of the practical cases) to avoid the computation of $q_t$. On the other hand, since $\widetilde{f}_t$ and $\widetilde{q}_t$ admits, respectively, $f_t$ and $q_t$ as marginals, Lemma \ref{lemma:marginalization} tells us the price we need to pay: \emph{an increase in the variance of the importance weights}. Fortunately, the same Lemma provides us some insights on how to optimally choose the backward kernels.
%The following Lemma is an important result towards the choice of the optimal backward kernel.

\begin{lemma} \label{lemma:marginalization}Let $f(\x_1,\x_2)$ and $g(\x_1,\x_2)$ be two probability densities with $supp(f) \subset supp(g)$. Then 
$$Var_g\left( \frac{f(\X_1, \X_2)}{g(\X_1,\X_2)} \right) \geq Var_g\left( \frac{f_1(\X_1)}{g_1(\X_1)}\right),$$
where $f_1(\x_1) = \int f(\x_1,\x_2) d\x_2$ and $g_1(\x_1) = \int g(\x_1,\x_2) d\x_2$.
\end{lemma}

\begin{proof}
From the variance decomposition, we have that
\begin{align*}
	Var_g\left( \frac{f(\X_1,\X_2)}{g(\X_1,\X_2)}\right) &= Var_g\left( \Exp_g\left[\frac{f(\X_1,\X_2)}{g(\X_1,\X_2)} \Big| \X_1 =\x_1 \right] \right) \\
	&+ \Exp_g\left[ Var_g\left(\frac{f(\X_1,\X_2)}{g(\X_1,\X_2)} \Big| \X_1=\x_1\right) \right] \\
	                                                 &\geq Var_g\left( \Exp_g\left[\frac{f(\X_1,\X_2)}{g(\X_1,\X_2)} \Big| \X_1= \x_1 \right] \right),
\end{align*}
since $f,g\geq 0$ (they are densities).

The result follows from the fact that the ratio of marginal densities can be rewritten as the following conditional expectation:
\begin{align*}
	\frac{f_1(\x_1)}{g_1(\x_1)} & = \int \frac{f(\x_1,\x_2)}{g_1(\x_1) g_{2|1}(\x_2|\x_1)} g_{2|1}(\x_2|\x_1) d\x_2 \\
													  & = \Exp_g\left( \frac{f(\X_1,\X_2)}{g(\X_1,\X_2)} \Big| \X_1 = \x_1 \right).
\end{align*}

\begin{flushright}
\qed
\end{flushright}

\end{proof}

As mentioned previously, Proposition \ref{prop:optimalBackwardKernel} shows how to design the backward kernels $\{L_{t-1}\}_{t=2}^T$ in order to minimize the variance of the importance weights.

\begin{proposition}[Optimal backward kernel] \label{prop:optimalBackwardKernel} The kernel  \\ $\displaystyle L_t^{opt}(\x_{t+1}, \x_t) := \frac{q_t(\x_t)K_{t+1}(\x_t,\x_{t+1})}{q_{t+1}(\x_{t+1})} $ is optimal in the sense that $Var_{\widetilde{q}_t}(w_t^{opt}(\X_{1:t})) \leq Var_{\widetilde{q}_t}(w_t(\X_{1:t}))$, where $\displaystyle  w_t^{opt}(\x_{1:t}) = \frac{f_t(\x_t)}{q_t(\x_t)}$.
\end{proposition}

\begin{proof} If we substitute the optimal kernel into the definition of importance weights we have that
\begin{align*}
	w^{opt}_t(\x_{1:t}) &= \frac{\widetilde{f}_t(\x_{1:t})}{\widetilde{q}_t(\x_{1:t})} \\
							 &= \frac{f_t(\x_t) \prod_{s=1}^{t-1}L^{opt}_s(\x_{s+1}, \x_s) }{ q_1(\x_1) \prod_{j=2}^t K_j(\x_{j-1}, \x_j)} \\ 
							 &= \frac{f_t(\x_t) \prod_{s=1}^{t-1} \frac{q_t(\x_t)K_{t+1}(\x_t,\x_{t+1})}{q_{t+1}(\x_{t+1})}    }{ q_1(\x_1) \prod_{j=2}^t K_j(\x_{j-1}, \x_j)} \\ 
							 &= \frac{f_t(\x_t)}{q_t(\x_t)}.
\end{align*}
The result, then, follows from Lemma \ref{lemma:marginalization}. 

\end{proof}

From Proposition \ref{prop:optimalBackwardKernel} we can see that if we know how to sample from $q_t(\x_t)$ then the SMC sampler algorithm reduces to a sequence of Importance Sampling steps, where at each time $t$ we sample particles from $q_t(\x_t)$ and correct the bias through the weights $w_t(\x_t) = \frac{f_t(\x_t)}{q_t(\x_t)}$.

\paragraph{\textbf{Independent kernel}}
One situation where we know how to calculate $q_t(\x_t)$ is when $K_t(\x_{t-1}, \x_t)$ does not depend on $\x_{t-1}$, making the mutation step completely memory-less. As an abuse of notation, let $K_t(\x_t) := K_t(\x_{t-1}, \x_t)$. This choice is not always recommended to be used in practice due to difficulties in designing an appropriate kernel. In this case it is easy to see that it is possible to perform a sequence of Importance Sampling steps, since
\begin{align*}
	q_t(\x_t) &= \int q_1(\x_1) \prod_{j=2}^t K_j(\x_{j-1}, \x_j) d\x_{1:t-1} \\
					 &= \int q_1(\x_1) \prod_{j=2}^t K_j(\x_j) d\x_{1:t-1} \\
					 &= \left( \int q_1(\x_1) d\x_1 \right) \left(\prod_{j=2}^{t-1}  \int K_j(\x_j) d\x_{1:t-1} \right) \left( \int K_t(\x_t) d\x_t \right) \\
					& = K_t(\x_t).
\end{align*}

\paragraph{\textbf{Approximations of the optimal kernel}}
Various approximations of the optimal backward kernel have been proposed in the literature (see, for example \cite{DelMoral2006}, Section 3.3.2) but here we will discuss only one of them.

If we rewrite the optimal backward kernel from Proposition \ref{prop:optimalBackwardKernel} as
$$L_t^{opt}(\x_{t+1}, \x_t) = \frac{q_t(\x_t) K_{t+1}(\x_t, \x_{t+1}) }{\int q_t(\x_t) K_{t+1}(\x_t, \x_{t+1}) d\x_t}$$
it suggests that a sensible approximation for this kernel is to use $\pi_t$ instead of $q_t$. In this case,
\begin{equation} \label{eq:approxOptimalBackwardKernel}
L_t^{opt}(\x_{t+1}, \x_t) \approx \frac{f_t(\x_t) K_{t+1}(\x_t, \x_{t+1}) }{\int f_t(\x_t) K_{t+1}(\x_t, \x_{t+1})d\x_t},
\end{equation}
since the normalizing constants of $\pi$ cancel out.

Although the integral in the denominator of (\ref{eq:approxOptimalBackwardKernel}) is (usually) not analytically tractable we can use the weighted sample $\{\x_{1:t}^{(j)}, w_t^{(j)} \}_{j=1}^N$ from $\pi_t$ generated by the SMC sampler procedure to approximate $L_t^{opt}$ as
\begin{equation} \label{eq:approxOptimalBackwardKernel2}
	L_t(\x_{t+1}, \x_t) = \frac{f_t(\x_t) K_{t+1}(\x_t, \x_{t+1}) }{\sum_{j=1}^N w_t^{(j)} K_{t+1}(\x_t^{(j)}, \x_{t+1}) }.
\end{equation}

\begin{figure}[ht]%
\includegraphics[width=1\textwidth, trim = 0mm 0.5cm 1cm 0mm]{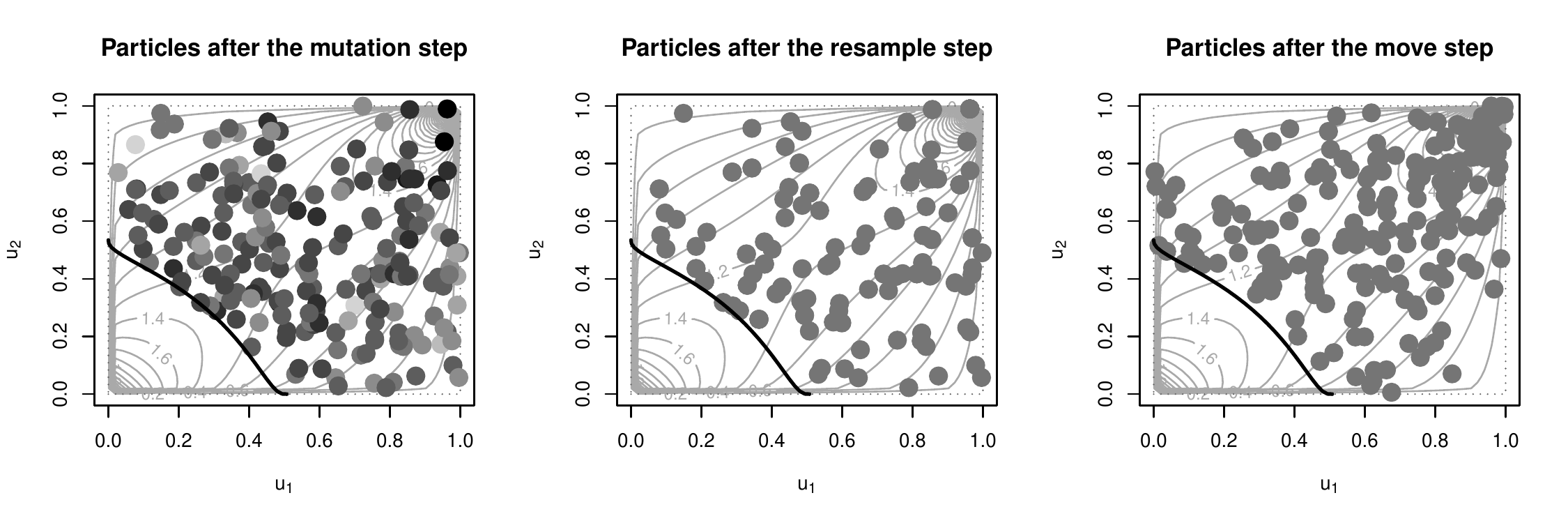}
\includegraphics[width=1\textwidth, trim = 0mm 0.5cm 1cm 0mm]{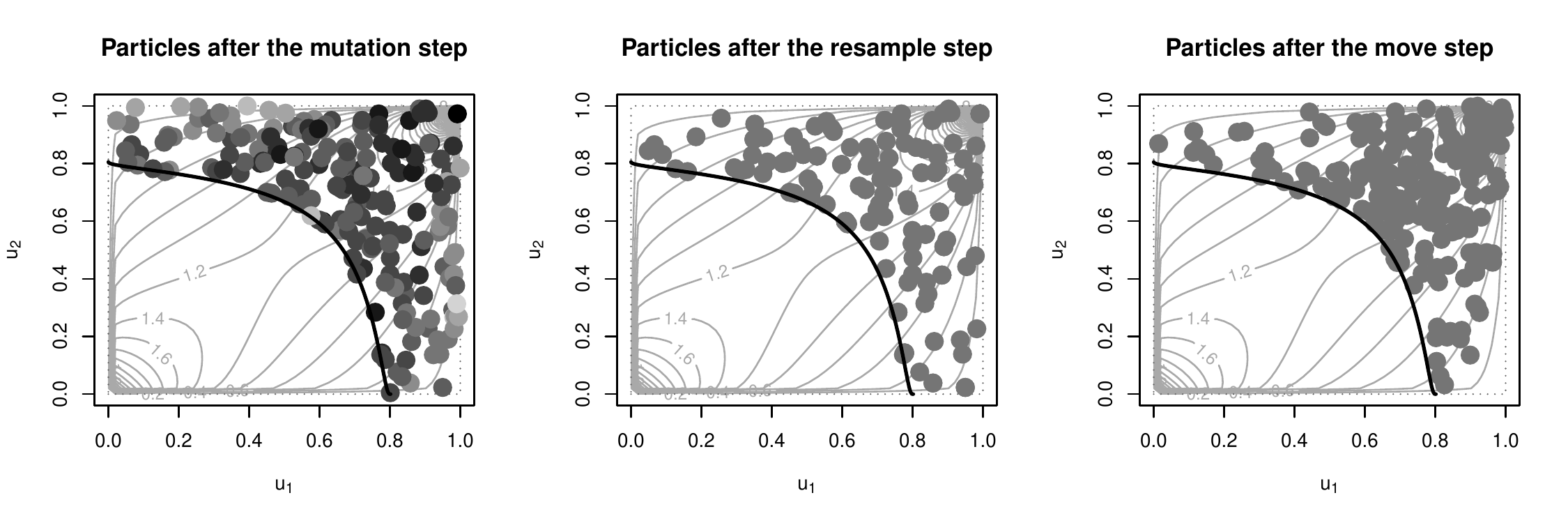}
\includegraphics[width=1\textwidth, trim = 0mm 0.5cm 1cm 0mm]{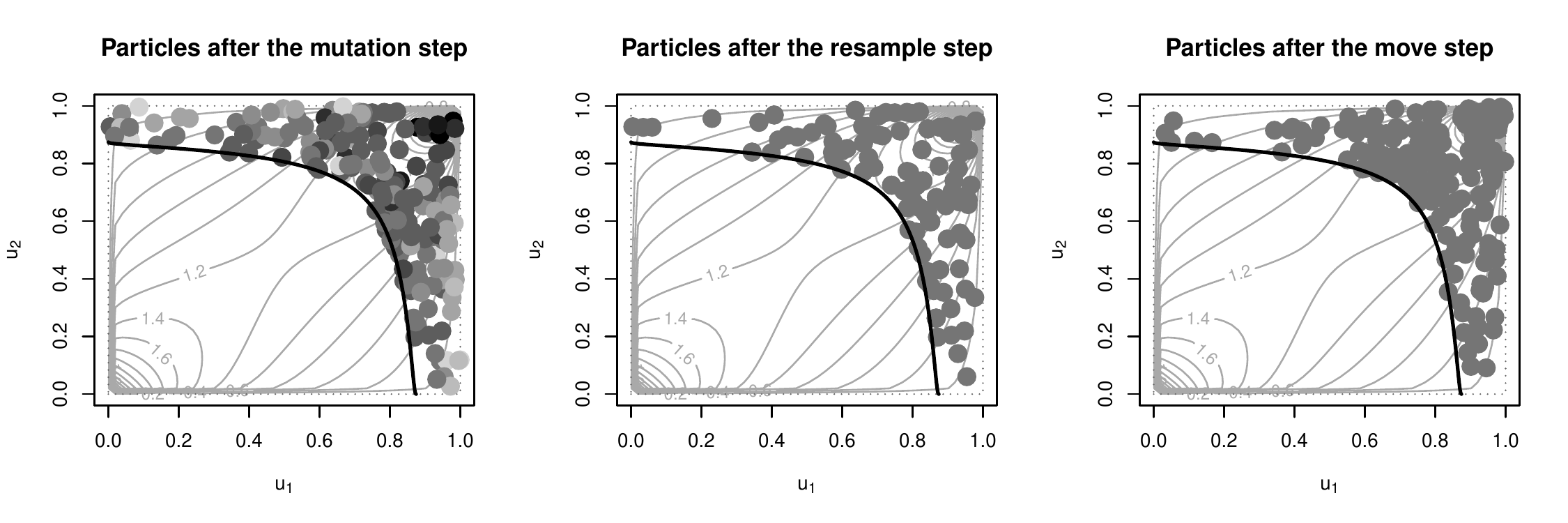}
\caption{\small Evolution of particles in the SMC algorithm, with marginals $X_{t,i} \sim LN(10-0.1i, \, 1+0.2i)$, $i=1,2$, $t=1,...,T$; Gumbel($\theta=1.5$); $N=200$ particles and $T=20$ equidistant levels from $B_1=0,...,B_T=100,000$; (top) $t=4$, (middle) $t=10$, (bottom) $t=14$. The level sets in the background represent the structure of the base, unconstrained copula.}
\end{figure}

%\FloatBarrier

%================================================================================================
\subsubsection{Forward kernels selection} \label{optimalForwardKernels}
%================================================================================================
So far we have discussed how to design backward kernels $L_t$ that are optimal for specific choices of forward kernels $K_t$. We now present possible choices of $K_t$ in order to have simple forms of importance weights when used with the optimal choice of backward kernel.

In this sense, one convenient choice for the forward kernels $K_t$ is to assume they are such that $K_t(\x_{t-1}, \x_t)$ has $\pi_t$ as invariant density, ie,
$$\pi_t(\x_{t+1}) = \int K_{t+1}(\x_{t}, \x_{t+1}) \pi_{t+1}(\x_{t}) d\x_{t}.$$
Proceeding this way, we can choose the backward kernel as follows,
\begin{equation}
L_t(\x_{t+1}, \x_t) = \frac{\pi_{t+1}(\x_t)K_{t+1}(\x_t, \x_{t+1})}{\pi_{t+1}(\x_{t+1})},
\label{eq:reversalKernel}
\end{equation}
which is a reasonable approximation for the optimal backward kernel from Proposition \ref{prop:optimalBackwardKernel}. It also worth noticing the kernel defined on (\ref{eq:reversalKernel}) is the reversal Markov kernel associated with $K_{t+1}$.

For theoretical purposes, we may also assume the forward kernel mixes perfectly, ie, $K_{t+1}(\x_t, \x_{t+1}) = \pi_{t+1}(\x_{t+1}).$ This choice of kernels is obviously not feasible in practice, since $\pi_{t+1}(\x_{t+1})$ is precisely the density we are trying to sample from, but it can provide us interesting insights. In this case, the incremental weights of the SMC sampler algorithm (see Algorithm \ref{algo:SMC_sampler}) are given by
\begin{align*}
	\alpha(\x_{t-1}, \x_{t}) &= \frac{f_t(\x_t) {L_{t-1}(\x_t, \x_{t-1}) } }{f_{t-1}(\x_{t-1})K_t(\x_{t-1},\x_t)} \propto %\\
	%&\propto \frac{\pi_t(\x_{t}) \pi_{t}(\x_{t-1})K_{t}(\x_{t-1}, \x_{t}) / \pi_{t}(\x_{t}) }{\pi_{t-1}(\x_{t-1})K_t(\x_{t-1},\x_t)} \\
	\frac{\pi_t(\x_{t-1})}{\pi_{t-1}(\x_{t-1})},
\end{align*}
which makes the weights at time $t$ independent of the particles sampled at time $t$.

%================================================================================================
\subsection{\textbf{Generalizing the relative error expression to the path-space IS and SMC Sampler}} \label{sec:variance}
%================================================================================================

As in Equation (\ref{eq:conditionalExpectation_u}), throughout this section we will assume we are interested in estimating conditional expectations such as $\Exp[h(\X) \, | \, \X \in \mathcal{G}_\X ] = \Exp[\varphi(\Ubold) \, | \, \Ubold \in \mathcal{G}_\Ubold]$ where $h(\X) = h\big(F_1^{(-1)}(u_1),...,F_d^{(-d)}(u_d)\big) =: \varphi(\Ubold)$. Moreover, let us denote by $p_t = \Prob[\X_t \in \mathcal{G}_{\X_t} ] = \Prob[ \Ubold_t \in \mathcal{G}_{\Ubold_t}]$. Throughout this section we will denote by $\Exp_{\pi_T}[ \cdot ]$ the expectation with respect to the density $\pi_T$.

In general, given $\widehat{\varphi}$ an estimate for $\Exp[\varphi(\Ubold) \, | \, \Ubold \in \mathcal{G}_\Ubold]$ a reasonable efficiency measure is given by its variance (the smaller the variance the better). In many cases, however, both the quantity of interest and the variance of the estimator are so small that using the variance as a measure efficiency is meaningless. For example, let us assume the quantities of interest are given by $z_s = \Prob[\Ubold \in A_s]$, for some sequence of sets $A_s \downarrow \emptyset$ when $s \rightarrow \infty$. It is trivial to see that if the estimator of the probability via simple Monte Carlo is denoted by $\widehat{\varphi_s}$ then $Var(\widehat{\varphi_s}) \rightarrow 0$ when $s \rightarrow \infty$.

In such scenarios, as discussed in \cite{asmussen2007stochastic}, Chapter VI, the relevant performance measure is to verify if the (unbiased) estimators possess bounded relative error, ie, if
$$\limsup_{s\rightarrow \infty} \frac{Var(\widehat{\varphi_s})}{z_s^2} < \infty.$$
In this section we present the generic variance expression for three different estimators, targeting $\Exp_{\pi_t}[\varphi(\Ubold_t)]$, namely: Importance Sampling (IS), Path space Importance Sampling (PIS) and Sequential Monte Carlo Sampler (SMCS). The reader should note, from (\ref{eq:asymptoticVar}), that the last two algorithm are \emph{asymptotically} unbiased, but may have some finite sample bias. For some specific quantities, \cite{arbenz2014importance} show that their estimator will produce bounded relative error and the same is valid for the Importance Sampling algorithm from \cite{mcleish2010bounded}, when targeting small probabilities.
%$\varphi(\Ubold) = \one_{A_s}(\Ubold)$ and $\mathcal{G}_\Ubold$ as discussed in \cite{asmussen2007stochastic}, Chapter VI, this may not be the best strategy when analyzing rare events

For SMC Samplers \cite{DelMoral2006} provide two Central Limit Theorem (CLT) results in ``extreme'' cases: the first one deals with the case when no resampling is performed and the second one for the case of resampling (using multinomial resample) at each iteration of the algorithm. These two cases will be briefly reviewed in the sequel, in order for us to have interpretable results for the asymptotic variances.

In both cases, under suitable integrability conditions discussed in \cite{DelMoral2006}, Proposition 2, the following convergence in distribution is valid when the number of particles tends to infinity,
\begin{equation}
\sqrt{N} \Big( \sum_{j=1}^N W_t^{(j)} \varphi(\Ubold^{(j)}) - \Exp_{\pi_t}[\varphi(\Ubold_t)]\Big) \underset{N\rightarrow + \infty}{\Longrightarrow} N(0, \, \sigma^2_t), \quad \forall t=1,...,T
\label{eq:asymptoticVar}
\end{equation}
where the normalized weights $W_t^{(j)}$ calculated in the SMC Sampler algorithm (see Algorithm \ref{algo:SMC_sampler}) and the limiting variance $\sigma^2_t$ are dependent on the resampling strategy.

\subsubsection{Importance Sampling (IS)}
Let us assume we use a density $q_0(\ubold_0)$ as an importance distribution targeting $\pi_T(\ubold) = \pi(\ubold)\one_{\{\ubold \in \mathcal{G}_{\Ubold_T} \}}(\ubold) / p_T$ (as in \ref{eq:target})).  If we denote the Importance Sampling estimator for $\varphi(\Ubold)$ by 
\begin{equation}
\widehat{\varphi}_{IS} : = \frac{1}{N}\sum_{j=1}^N \varphi(\ubold^{(i)}_0) w_{IS}(\ubold^{(i)}_0),
\label{eq:IS_est}
\end{equation}
where $w_{IS}(\ubold_0):= w_{IS}(\ubold_0) = \frac{\pi_T(\ubold_0)}{q_0(\ubold_0)}$ then we have that
\begin{equation*}
	Var(\widehat{\varphi}_{IS}) = \frac{1}{N}\left[\int \frac{\pi_T^2(\ubold)}{q_0^2(\ubold)}\varphi(\ubold) q_0(\ubold) d\ubold - \left( \int \frac{\pi_T(\ubold)}{q_0(\ubold)}\varphi(\ubold) q_0(\ubold) d\ubold \right)^2\right]. %\\
	%&= \frac{1}{N}\left[\int \frac{\pi_T(\ubold)}{q_0(\ubold)}\varphi(\ubold) \pi_T(\ubold) d\ubold - \Exp_{\pi_T}^2[\varphi(\Ubold)]\right].
\end{equation*}
In particular, if $\pi_T \geq q_0$ $\pi_T$-almost surely (eg, the IS distribution is the unconditional distribution $\pi$), then
\begin{equation}
Var(\widehat{\varphi}_{IS}) \geq \frac{1}{N} Var_{\pi_T}(\varphi(\Ubold)).
\label{eq:var_IS}
\end{equation}

In the capital allocation problem we know the probability of the conditioning event is chosen to be $p_T = 1-\alpha$, meaning that the target distribution $\pi_T$ will be known (including its proportionality constant) if $\pi$ is perfectly known. In this scenario, if the importance distribution is the unconditional density, ie, $q_0 \equiv \pi$ then the importance weights can be perfectly calculated. However, in many interesting cases either $\pi$ or $q_0$ may only be known up to a constant. In this case it is necessary to use the ``self-normalized'' version of the importance weights, namely,
\begin{equation}
w_{IS\text{-}SN}^{(i)} = \frac{\pi_T(\ubold^{(i)}) / q_0(\ubold^{(i)}) }{\sum_{j=1}^N \pi_T(\ubold^{(j)}) / q_0(\ubold^{(j)})}.
\label{eq:IS_self_normalized}
\end{equation}

When the self-normalized importance weights are used to estimate $\varphi(\Ubold)$ it is easy to show the variance of the new estimator will be larger than the one defined in (\ref{eq:IS_est}). More formally, if we define $\widehat{\varphi}_{IS\text{-}SN} : = \sum_{j=1}^N \varphi(\ubold^{(i)}_0) w_{IS\text{-}SN}(\ubold^{(i)}_0),$ then 
$$\frac{1}{N} Var_{\pi_T}(\varphi(\Ubold)) \leq Var(\widehat{\varphi}_{IS}) \leq Var(\widehat{\varphi}_{IS\text{-}SN}).$$
\subsubsection{Path space Importance Sampling (PIS)}
When no resampling is performed in the SMC Sampler algorithm, it is easy to see that, at each time step, it collapses to an Importance Sampling algorithm in the Path space (PIS). Since the particles have no interaction, the (unnormalized) weights in (\ref{eq:asymptoticVar}) will be given by
$$w_{PIS,t}(\Ubold_t) = \frac{\widetilde{\pi}_t(\ubold_{1:t})}{\widetilde{q}_t(\ubold_{1:t})},$$
where $\widetilde{q}_t(\ubold_{1:t}) =\widetilde{q}_1(\ubold_1)\prod_{j=2}^t K_j(\ubold_{j-1},\ubold_j)$ and $\widetilde{\pi}_t(\ubold_{1:t}) = \pi_t(\ubold_t) \prod_{s=1}^{t-1}L_s(\ubold_{s+1}, \ubold_s)$. In this case the following result is valid.

\begin{proposition}[Path Importance Sampling (PIS) CLT] Under the integrability assumptions given in \cite{DelMoral2006}, Proposition 2, the asymptotic variance, as the number of particles increases to infinity, in (\ref{eq:asymptoticVar}) is given by
\begin{equation}
\sigma_{t}^2 = \sigma_{PIS,t}^2 := \int \frac{\widetilde{\pi}_t^2(\ubold_{1:t})}{\widetilde{q}_t(\ubold_{1:t})}\{\varphi(\ubold_{t}) - \Exp_{\pi_t}[\varphi(\Ubold_t)] \}^2 d\ubold_{1:t}
\label{eq:asymptoticVar_PIS}
\end{equation}

\end{proposition}

In the particular case where the forward kernel is perfectly mixing and the backward kernels are given by (\ref{eq:reversalKernel}) then it is easy to show that
\begin{equation}
\sigma_{PIS,t}^2 = \int \frac{\pi_1^2(\ubold_{0})}{q_{0}(\ubold_{0})} d\ubold_0 \prod_{k=2}^t \int \frac{\pi_k^2(\ubold_{k-1})}{\pi_{k-1}(\ubold_{k-1})} d\ubold_{k-1} Var_{\pi_t}(\varphi(\Ubold)).
\label{eq:asymptoticVar_PIS_2}
\end{equation}

Moreover, let us assume the sequence of distributions is given by (\ref{eq:target}). Then, each integral can be simplified to
\begin{align*}
\int \frac{\pi_k^2(\ubold_{k-1})}{\pi_{k-1}(\ubold_{k-1})} d\ubold_{k-1} & = \int \frac{\pi^2(\ubold_{k-1}) \one_{\{\ubold_{k-1} \in \mathcal{G}_{\Ubold_{k}}\}}(\ubold_{k-1})}{\pi(\ubold_{k-1}) \one_{\{\ubold_{k-1} \in \mathcal{G}_{\Ubold_{k-1}}\}}(\ubold_{k-1})} \frac{p_{k-1}}{p_k^2} d\ubold_{k-1} \\
&= \frac{p_{k-1}}{p_k} \int \frac{\pi(\ubold_{k-1}) \one_{\{\ubold_{k-1} \in \mathcal{G}_{\Ubold_{k}}\}}(\ubold_{k-1})}{p_k} d\ubold_{k-1} \\
& = \frac{p_{k-1}}{p_k}.
\end{align*}

But since the events become rarer the larger the index $k$, then $\frac{p_{k-1}}{p_k} > 1$
(and the same argument holds for the first integral, if, as in the IS algorithm, we start sampling from the unconditional distribution) and for $N$ sufficiently large, 
$$Var(\widehat{\varphi}_{PIS,T}) \approx \frac{1}{N} \sigma_{PIS,t}^2 > \frac{1}{N} Var_{\pi_T}(\varphi(\Ubold)).$$

\subsubsection{Sequential Monte Carlo Sampler (SMCS) -- resampling at every time step}
When resampling is performed at every step of the SMC algorithm we can still calculate the variance $\sigma^2$ from (\ref{eq:asymptoticVar}) through the Proposition \ref{eq:asymptoticVar_AR_SMCS}. The estimate calculated through this scheme will be denoted here $\widehat{\varphi}_{SMCS}$.

\begin{proposition}[Always-Resampling SMC Sampler (SMCS) CLT] Under the integrability assumptions given in \cite{DelMoral2006}, Proposition 2, the asymptotic variance in (\ref{eq:asymptoticVar}) is given by
\begin{equation}
{\small
\begin{split}
\sigma_{t}^2 = &\sigma_{SMCS,t}^2 := \int \frac{\widetilde{\pi}^2_t(\ubold_{1})}{\mu_t(\ubold_{1})}\left(\int \varphi(\ubold_t)\widetilde{\pi}\left(\ubold_t|\ubold_1\right)d\ubold_t - E_{\pi_t}(\varphi(\Ubold)) \right)^2 d\ubold_{1}\\
&+ \sum_{k=2}^{t-1}\int \frac{\left( \widetilde{\pi}_t(\ubold_k) L_{k-1}(\ubold_k,\ubold_{k-1}) \right)^2}{\left( \pi_{k-1}(\ubold_{k-1}) K_{k}(\ubold_{k-1},\ubold_k) \right)}\left(\int \varphi(\ubold_t)\widetilde{\pi}\left(\ubold_t|\ubold_k\right)d\ubold_t - E_{\pi_t}(\varphi(\Ubold)) \right)^2 d\ubold_{k-1}d\ubold_k\\
&+ \int \frac{\left( \pi_t(\ubold_t) L_{t-1}(\ubold_t,\ubold_{t-1}) \right)^2}{\left( \pi_{t-1}(\ubold_{t-1}) K_{t}(\ubold_{t-1},\ubold_t) \right)}\left(\varphi(\ubold_t) - E_{\pi_t}(\varphi) \right)^2 d\ubold_{t-1}d\ubold_t.,
\end{split}
}
\label{eq:asymptoticVar_AR_SMCS}
\end{equation}
where \small{$\widetilde{\pi}_n(\ubold_k):=\int \widetilde{\pi}_n(\ubold_{1:n})d\ubold_{1:k-1}d\ubold_{k+1:n}$} and \small{$\widetilde{\pi}_n(\ubold_n|\ubold_k):=\int \widetilde{\pi}_n(\ubold_{1:n})d\ubold_{1:k-1}d\ubold_{k+1:n}/\widetilde{\pi}(\ubold_k)$}.

\end{proposition}

As in the previous section, if we assume the mutation kernel is perfectly mixing and if we use the approximation of the optimal kernel given by (\ref{eq:reversalKernel}) then the variance in (\ref{eq:asymptoticVar_AR_SMCS}) can be simplified to
\begin{equation}
\sigma_{SMCS,t}^2 = \int \frac{\pi_t^2(\ubold_{t-1})}{\pi_{t-1}(\ubold_{t-1})} d\ubold_{t-1} Var_{\pi_t}(\varphi(\Ubold)),
\label{eq:asymptoticVar_AR_SMCS_2}
\end{equation}
from which it becomes clear from (\ref{eq:asymptoticVar_PIS_2}) and (\ref{eq:asymptoticVar_AR_SMCS_2}) that $\sigma_{PIS,t}^2 > \sigma_{SMCS,t}^2 $, which implies that
$$Var(\widehat{\varphi}_{SMCS,T}) < Var(\widehat{\varphi}_{PIS,T}) .$$ %< \frac{1}{N} Var_{\pi_T}(\varphi(\Ubold)).$$

%{\textbf{\color{red}The idea was to prove that
%$$Var(\widehat{\varphi}_{SMCS,T}) < Var(\widehat{\varphi}_{PIS,T}) < Var_{\pi_T}(\varphi(\Ubold)) < Var(\widehat{\varphi}_{IS}) < Var(\widehat{\varphi}_{IS\text{-}SN}),$$
%but the second inequality didn't work out...
%}}

%================================================================================================
%================================================================================================
\section{Design of a SMC sampler with linear constraints for capital allocation} \label{sec:linearConstraints}
%================================================================================================
%================================================================================================
In this section we return to the problem of sampling from the distribution of \\ $(\X \, \big| \, \sum_{i=1}^d X_i > B)$ producing samples from $\textbf{U} \in \mathcal{G}_{\textbf{U}}$, as explained in section \ref{sec:simForConditionalExp}. To use the algorithm specified in Section \ref{sec:SMCsamplers} we still need to design: (1) the forward kernels $K_t(u_{t-1}, u_t)$, (2) the backward kernels $L_{t-1}(u_t, u_{t-1})$ and (3) a Markov Chain move kernel $M$ (in the spirit of Section \ref{ResampleMove}). For the backward kernel we will use the approximation to the optimal one presented in section \ref{optimalBackwardKernels}; the forward kernel and the ``move'' kernel will be presented, respectively in sections \ref{sec:forward_kernel_linear} and \ref{sec:moveKernel}

\subsection{\textbf{Forward Kernel}} \label{sec:forward_kernel_linear}
For the forward kernel $K_t(\ubold_{t-1}, \ubold_t),$ if we can guarantee that any move from $\ubold_{t-1}$ to $\ubold_t$ will already be in $\mathcal{G}_{\textbf{U}_t}$ then we will not loose any particle in the mutation step, improving the efficiency of the algorithm. Since we are developing the sampling procedure in $[0,1]^d$ then, under the assumption that we can precisely characterize the constraint region $\mathcal{G}_{\textbf{U}_t}$ (ie, we can calculate $F_i^{-1}$ for all $i=1,...,d$) then we can propose a ``slice-sampling'' procedure for $K_t$.

The idea of this type of kernel is that we can first sample $d-1$ coordinates of the $\textbf{u}_t$ vector (chosen randomly) and then, conditional on these values, sample the last component constrained to an interval that will ensure that $\sum_{i=1}^d x_i > B_t$. In general these kernels will look like 
\begin{equation}
K_t(\textbf{u}_{t-1}, \, \textbf{u}_t) = \sum_{m=1}^d \left[ K_t^{(-m)}(\textbf{u}_{t-1,-m}, \textbf{u}_{t,-m}) K_t^{(m)}(u_{t-1,m}, u_{t,m}) \right] p_m,
\label{eq:mixtureMutationKernel}
\end{equation}
where $p_m$ is the probability of the $m$-th coordinate being the last one to be chosen,\\ $\textbf{u}_{t,-m}= (u_{t,1},...,u_{t,m-1},u_{t,m+1},...,u_{t,d})$ is the vector $\ubold_t$ without its $m$-th coordinate and $K_t^{(-m)}$ is the kernel that moves the $d-1$ dimensions of $\textbf{u}_{t-1,-m}$ to time $t$. Similarly, $K_t^{(m)}$ denotes the kernel that moves $u_{t-1,m}$ to $u_{t,m}$ ensuring that $\sum_{i=1}^d x_{t,i} > B_t$.

To guarantee the condition is satisfied, $K_t^{(m)}$ needs to be defined over $[B_t^u(m), 1]$, where 
\begin{equation}
B_t^u(m):= F_m^{-1}(B_t^x(m))
\label{eq:B_u}
\end{equation}
with
\begin{equation}
B_t^x(m) := \max\Bigg\{0, \,  B_t - \sum\limits_{\substack{i=1 \\ i\neq m}}^{d} F_i(u_{t,i}) \Bigg\}.
\label{eq:B_x}
\end{equation}

For simplicity, we can choose the last move to be uniformly distributed in $[B_t^u(m), 1]$, leading to
$$K_t^{(m)}(u_{t-1,m}, u_{t,m}) =  \frac{u_{t,m}}{1 - B_t^u(m)} \one_{\{ u_{t,m} \in [B_t^u(m), 1]\}} (u_{t,m}).$$

Again, for the sake of simplicity, we will only discuss the case where $K_t^{(-m)}$ consists of independent moves in each dimension, ie,
\begin{equation}
K_t^{(-m)}(\textbf{u}_{t-1}, \textbf{u}_t) = \prod\limits_{\substack{i=1 \\ i\neq m}}^{d}K_t^{(-m,i)}(u_{t-1,i}, u_{t,i}).
\label{eq:marginalMutationKernel}
\end{equation}
Moreover, it will be assumed that $p_m=1/d$, for all $m=1,...,d$.

\paragraph{\textbf{Uniform moves in} $\mathcal{G}_{\textbf{U}}$ }
The first (na\"ive) idea is to define the move in each component of $\textbf{u}$ as uniform, leading to a marginal kernels $K_t^{(-m,i)}(u_{t-1,i},  \, u_{t,i}) = u_{t,i} \one_{\{u_{t,i} \in [0,1] \}}(u_{t,i})$ and, consequently,
$$K_t(\textbf{u}_{t-1}, \textbf{u}_t) = \frac{1}{d} \sum_{m=1}^d \left( \prod\limits_{\substack{i=1 \\ i\neq m}}^{d} u_{t,i} \one_{\{u_{t,i} \in [0,1] \}}(u_{t,i})\right) \left( \frac{u_{t,m}}{1 - B_t^u(m)} \one_{\{ u_{t,m} \in [B_t^u(m), 1]\}} (u_{t,m})\right).$$
 
As we can see from the construction of this kernel, it is clearly independent of $u_{t-1}$ and the comments by the end of Section \ref{optimalBackwardKernels} apply, meaning that the problem reduces to a series of Importance Sampling problems.

\paragraph{\textbf{Global adaptive Beta moves in} $[0,1]^{d-1}$ }
One strategy to use the information contained in $\textbf{u}_{t-1}$ in the mutation step is to use the whole set of weighted particles at $t-1$ to estimate the parameters of the mutation kernel (subject to some restriction).

Since our kernels are defined in $[0,1]$ a reasonable idea is to use a global Beta kernel, in the sense that all particles at time $t-1$ will be mutated through the same kernel. To select the parameters of the Beta distribution we match the first two moments of the Beta distributions with its sample moments at time $t-1$. Formally, let us denote $\{ \textbf{u}_{t-1}^{(j)}, \, W_{t-1}^{(j)} \}_{j=1}^N$ the set of $N$ weighted particles at time $t-1$ and $K_t^{(-m, i)}(u_{t-1,i}, \, u_{t,i}) = Beta(u_{t,i}; \, \alpha_{t-1,i}, \, \beta_{t-1,i})$, where the RHS term denotes the density of a random variable $Y_{t-1,i}$ Beta distributed with parameters $\alpha_{t-1,i}$ and $\beta_{t-1,i}$ evaluated at $u_{t,i}$. Then matching the first two moments we have
{\small
\begin{align*}
\Exp[Y_{t-1,i}] &= \frac{\alpha_{t-1,i}}{\alpha_{t-1,i} + \beta_{t-1,i} } = \sum_{j=1}^N W_{t-1}^{(j)} u_{t-1,i}^{(j)} =: \widehat{\mu}_{t-1,i}, \\
	Var(Y_{t-1,i}) & = \frac{  \alpha_{t-1,i}\beta_{t-1,i}   }{   ( \alpha_{t-1,i} + \beta_{t-1,i})^2 (\alpha_{t-1,i} + \beta_{t-1,i} +1)   } \approx \widehat{\mu}^2_{t-1,i}  - \sum_{j=1}^N W_{t-1} \left( u_{t-1,i}^{(j)} \right)^2 : = \widehat{\sigma^2}_{t-1,i}
\end{align*}
}
and after some algebra we find
\begin{align*}
\widehat{\alpha}_{t-1,i} &=  \frac{(1- \widehat{\mu}_{t-1,i} )}{\widehat{\sigma^2}_{t-1,i} } \widehat{\mu}_{t-1,i}^2 \\
\widehat{\beta}_{t-1,i} &= \widehat{\alpha}_{t-1,i}\left(\frac{1}{\widehat{\mu}_{t-1,i}} -1 \right).
\end{align*}

Therefore an approximation for the mutation kernel at time $t$ and dimension $i$ -- as in (\ref{eq:marginalMutationKernel}) -- is given by
\begin{equation}
K_t^{(-m,i)}(u_{t-1,i}, \, u_{t,i}) \approx Beta(u_{t,i}; \, \widehat{\alpha}_{t-1,i}, \, \widehat{\beta}_{t-1,i}).
\label{eq:globalBetaKernel}
\end{equation}

\begin{remark} It is important to emphasize the kernel in (\ref{eq:globalBetaKernel}), when plugged into (\ref{eq:marginalMutationKernel}) and (\ref{eq:mixtureMutationKernel}), does not require any tuning and is not independent of $u_{t-1,i}$, since the parameters of the Beta distribution depend on these values.
\end{remark}

\begin{figure}%
\includegraphics[width=1\textwidth]{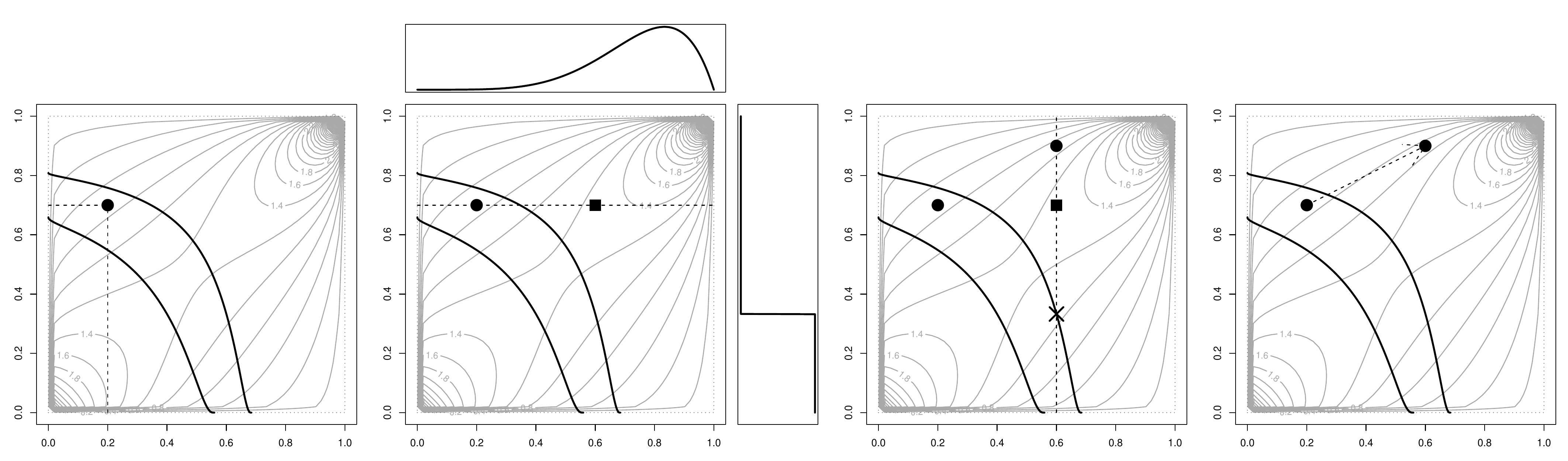}
\caption{\small Example of the Global Beta kernel for a Gumbel$(1.5)$ copula with Log-Normal marginals: $(\mu_1 = 0.6, \sigma_1 = 1.4)$, $(\mu_2=0.4, \sigma_2 =1)$. The boundary are such that $F_1^{-1}(u_1) + F_2^{-1}(u_2) > 2.25 \text{ and } 3.57$.}
\label{fig:beta_mutation}
\end{figure}

Figure \ref{fig:beta_mutation} exemplifies the mutation of one particle $u_{t-1} = (0.2,0.7)$ (which is in the $(t-1)$-th level set) to $u_t = (0.6,0.9)$ (which is in the $t$-th level set). The mutation starts moving the first coordinate of $u_{t-1}$ through a Beta distribution and then the second coordinate is moved following a uniform distribution defined in the appropriate region.

\subsection{\textbf{Markov Chain move kernel}} \label{sec:moveKernel}
Since the forward kernels designed in Section \ref{sec:forward_kernel_linear} ensure the new particles will satisfy the condition at level $t$ one possibility is to use the same kernel as a proposal in a M-H algorithm. The drawback would be that in higher dimensions the acceptance rate of the M-H would be extremely small. Instead, here we propose the usage of a Gibbs sampling algorithm, that should be always preferred when the full conditional densities are known.

Suppressing the dependence in $t$ in the vector $\textbf{u}$ and denoting $v^*(m) :=(u_1^*,..., u_m^*,$ $u_{m+1},...,u_{d})$ we have that the full conditional for a generic $m=1,...,d$ can be written as
\begin{align*}
	\pi_t(u_m^* \, | \, u_1^*,...,u_{m-1}^*,u_{m+1},...,u_d) & = \frac{\pi_t(u_1^*,...,u_m^*,u_{m+1},...,u_{d})}{\pi_t( u_1,...,u_{m-1},u_{m+1},...,u_d )} \\
	& \propto \pi_t(v^*(m)) \\
	& \propto c(v^*(m)) \one_{\left\{v^*(m) \in \mathcal{G}_{\textbf{U}_t} \right\}}(v^*(m)).
\end{align*}
Note that the full conditional distribution for the $m$-th coordinate of $\textbf{u}$ is a probability distribution for $u_m^*$. On the other hand, since $u_1^*,...,u_{m-1}^*,u_{m+1},...,u_d$ are fixed, we can rewrite the condition $[ v^*(m) \in \mathcal{G}_{\textbf{U}_t} ] $ as $u_m^* \in [B^u(m), 1]$, with $B^u(m)$ as in (\ref{eq:B_u}).

To sample $u_m^*$ from its full conditional distribution one can use a univariate slice sampler algorithm (see \cite{neal2003slice}), which only requires the full conditional target up to a normalizing constant. In Figure \ref{fig:slice_sampler} we present an example of such a Markov Chain move. On the RHS, the initial point is $(u_1, \, u_2) = (0.9, \, 0.3)$. First the support of the full conditional distribution is calculated, ie, $B^u(1) = 0.6$ and plotted as a cross on the second figure. Then, a value $u_1^*=0.8$ is sampled from $\pi(u_1 \, | \, u_2=0.3)$ (a square in the second plot). On top of the second plot we present the full conditional distribution (truncated on the left at $B^u(1)=0.6$. For this value we find that $B^u(2) = 0$ and the support of the next full conditional distribution is $[0,1]$ (the actual density is plotted vertically). The second coordinate $u_2^* =0.7$ is then sampled from $\pi(u_2 \, | \, u_1^*=0.8)$. In the last plot we have the final move, from $(u_1, \, u_2) = (0.9, \, 0.3)$ to $(u_1^*, \, u_2^*) = (0.8, \, 0.6)$.

\begin{figure}%
\includegraphics[width=1\textwidth]{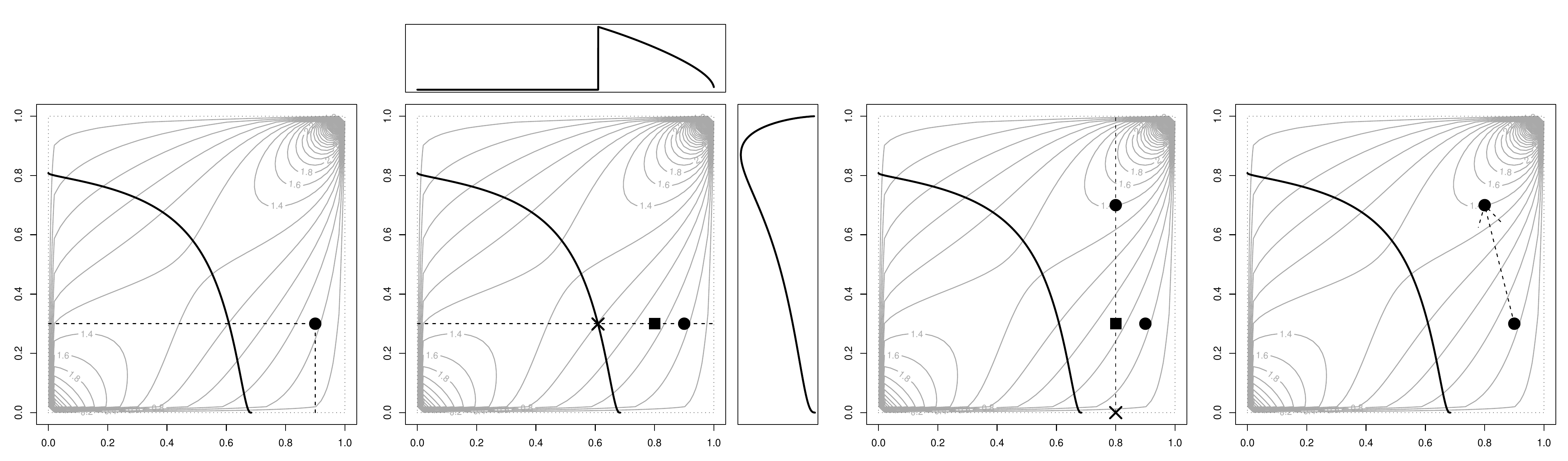}
\caption{\small Example of the move kernel for a Gumbel$(1.5)$ copula with Log-Normal marginals: $(\mu_1 = 0.6, \sigma_1 = 1.4)$, $(\mu_2=0.4, \sigma_2 =1)$. The boundary is such that $F_1^{-1}(u_1) + F_2^{-1}(u_2) > 3.57$.}
\label{fig:slice_sampler}
\end{figure}

%===========================================================================
%===========================================================================
%===========================================================================
\section{Case Studies}
\label{sec:Examples}
In this section we present some simulation examples of the performance of the proposed Copula-Constrained SMC Sampler algorithm. For all the simulations and density calculations we made extensive use of the {\ttfamily R}-package {\ttfamily copula} \cite{hofert2014copula}.

The aim of all methods presented here will be to calculate conditional expectations of the form
\begin{equation}
\Exp\left[X_k \, | \, \sum_{i=1}^d X_i > VaR_\alpha(\sum_{i=1}^d X_i) \right], \quad \text{ for } k=1,...,d
\label{eq:target_exp}
\end{equation}
where the $\alpha$ quantile ($VaR_\alpha$) is assumed to be perfectly known (see comments bellow) and \\
$\alpha \in (0.1, 0.2, ..., 0.9, 0.95, 0.99, 0.995, 0.999, 0.9995, 0.9999, 0.99995)$. Moreover, as in \cite{arbenz2014importance}, we will assume the marginal distributions of $\textbf{X}$ are Log-Normal with
$$X_i \sim LN(10-0.1i, \, 1+0.2i), \quad i=1,...,d.$$

For all the examples presented here, since we are not able to express the $VaR_\alpha$ of the aggregated process in a closed form, the first step is to to calculate a reliable approximation of this quantity, for each level $\alpha$ of interest. This is done through a Monte Carlo simulation of the loss vector $\X = (X_1,...,X_d)$ -- from which we can compute the aggregate loss -- of fixed size $N_q = 10,000,000$. Given a particular sample of size $N_q$, all the quantiles of the aggregate loss can be calculated. This process is, then, repeated for $N_{rep}= 500$ times, and the $\alpha$-quantile is set as the average of the $\alpha$-quantiles over all the $500$ runs. The reader should note that the estimate of extreme quantiles (for example $\alpha = 0.9999$) will be less precise than the estimate of lower quantiles (such as $\alpha = 0.3$), but for the purpose of comparing the proposed algorithm with competing strategies this is irrelevant, as long as the quantile used in the conditioning argument of (\ref{eq:target_exp}) is the same for all methods.

%\textcolor[rgb]{1,0,0}{\textbf{In a more advanced implementation of the SMC algorithm, the intermediate level sets can be estimated ``on-the-fly" (see ...) }}

After calculating the quantiles for all levels $\alpha$ (which are, from now on, assumed to be exact) the baseline comparison values for the expectations in (\ref{eq:target_exp}) are calculated as follows. For each level $\alpha$ we sample as many loss vectors $\X$ as necessary in order to have a sample of size $N_{MC} = 1,000$ satisfying the condition $\sum_{j=1}^d X_i > VaR_\alpha(\sum_{i=1}^d X_i)$. At this point we note that this naive Monte Carlo sampling strategy is very inefficient and would never be utilized in practice, due to the huge computational cost, but it provides us reference comparison to our more efficient SMC Sampler. To perform these simulations we required the usage of hundreds of cores from UCL Legion High Performance Computing Facility.

The expectations in (\ref{eq:target_exp}) are, then, estimated as
$$\widehat{\varphi}_{MC} = \frac{1}{N_{MC}} \sum_{j=1}^{N_{MC}} X_i^{(j)},$$
where $\X^{(j)}, \, j=1,...,N_{MC}$ are the samples satisfying the condition. This procedure is repeated $N_{rep}=500$ and the Monte Carlo estimate $\overline{\widehat{\varphi}}_{MC}$ is taken as the average over these $N_{rep}$ repetitions, as follows
$$\overline{\widehat{\varphi}}_{MC} = \frac{1}{N_{rep}} \sum_{k=1}^{N_{rep}} \widehat{\varphi}_{MC}^{(k)},$$
where $\widehat{\varphi}_{MC}^{(k)}$ stands for the estimate (using $N_{MC}$ particles) from the $k$-th run (out of $N_{rep}$). Analogously we can also define the variance of the MC estimator, $Var(\widehat{\varphi}_{MC})$. Moreover, we will denote by $N_{SMC}$ the number of particles used in the SMC algorithm, and by $Var(\widehat{\varphi}_{SMC})$ the variance of its estimate, also calculated using $N_{rep}$ runs.

One may observe that the expected number of samples in the Monte Carlo scheme in order to have $N_{MC}$ samples satisfying the $\alpha$ condition is equal to $M_{MC} = N_{MC}/(1-\alpha)$, which can be prohibitive if $\alpha$ is very close to 1.

For all the examples, the efficiency of the algorithm will be measured with respect to the Variance Reduction when compared with a simple Monte Carlo scheme (properly normalized). More formally, if the SMC algorithm uses $T$ levels to approximate (\ref{eq:target_exp}) then we will denote by Variance Reduction the ratio:
\begin{equation}
\text{Variance Reduction} = N_{MC} \times Var(\widehat{\varphi}_{MC}) \Bigg/ T \times N_{SMC} \times Var(\widehat{\varphi}_{SMC}).
\label{eq:variance_reduction}
\end{equation}
We note this is a conservative measure of the Variance Reduction, as typically practitioners may only use in the denominator $N_{SMC} \times Var(\widehat{\varphi}_{SMC})$. In addition, the Variance Reduction must be analysed in conjunction with the estimation bias. For this purpose we will study the Relative Bias, defined as the relative difference from the SMC estimate to the MC estimate (assumed to be the truth, due to the very large sample sizes taken):
$$\text{Relative Bias} = \frac{\widehat{\varphi}_{SMC} - \widehat{\varphi}_{MC}}{\widehat{\varphi}_{MC}}$$

If the level of interest of the expectations in (\ref{eq:target_exp}) is, for example, $\alpha =0.999$ then, the SMC algorithm designed here will use as intermediate levels the quantiles $\alpha =0.1, 0.2, ..., 0.9, 0.95, 0.99, 0.995$. Although expectations conditional on quantiles at lower levels, such as $0.1,...,0.9$ are not of direct interest for risk managers, as a by-product of the SMC algorithm, weighted samples from all the intermediate levels will be created and all the conditional expectations can be estimated.

%===========================================================================
%===========================================================================
%===========================================================================
\subsection{\textbf{Clayton copula dependence between risk cells}} \label{sec:clayton}
In this first study we assume we have a simple business unit and risk cell structure in which it is assumed that the dependence is on the annual losses and given simply by a Clayton copula. We first study a representative simple case, with dimension $d=5$ (see Definition \ref{def:arch}) and investigate the behaviour of the proposed algorithm for different parameters values with fixed number of particles $N_{SMC}=250$. To choose the parameters of interest we set the multivariate coefficient of (lower) tail dependence $\lambda_l$ (see \cite{de2012multivariate} or Definition \ref{def:mult_tail_dep} and Figure \ref{fig:clayton_lower_tail_dep}) to be approximately equals to $0.25, \, 0.5, \, 0.75 \text{ and } 0.9$, which led to $\theta= 0.16, \, 0.33, \, 0.78 \text{ and } 2.12$ (see Figure \ref{fig:clayton_lower_tail_dep}). To compare with the results presented in \cite{arbenz2014importance}, the parameter value $\theta = 1$ is also considered.

On Figure \ref{fig:clayton_bias_var}, we present the Relative Bias (top row) and Variance Reduction (bottom row) for all the range of different copula parameters and quantile levels. For ease of presentation, only three expectations are shown, namely the first marginal ($i=1$), the last marginal ($i=5$) and the sum of all marginal expectations (which is precisely the Expected Shortfall for the aggregated loss). From the Relative Bias analysis one can see that, regardless of the copula parameter and quantile levels, the estimation error is always smaller than $4\%$ (in absolute value).

Since the estimates are unbiased, it makes sense to look at the Variance Reduction set of plots (bottom row of Figure \ref{fig:clayton_bias_var}). In the vertical axis of the plot the $\log_{10}$(Variance Reduction) is presented, meaning that, for example, the variance of the SMC algorithm is $10^{1.17} \approx 15$ times smaller than the MC scheme when $\theta=2.12$ and $\alpha=0.999$. The horizontal line at $0$ defines the threshold where the SMC method outperforms a simple Monte Carlo: when the Variance Reduction is bellow the line the MC variance is smaller. As one should expect, for lower quantile levels a simple MC scheme should be preferred over the SMC method, since the condition in (\ref{eq:target_exp}) can be easily satisfied with a reasonably small sample size. On the other hand, as soon as the conditioning event becomes rarer, the variance of the simple Monte Carlo scheme starts to increase polynomially fast when compared with the variance of the proposed SMC algorithm.

The rarer the conditioning event the more computationally efficient it becomes to use the SMC Sampler method proposed. As previously mentioned, when $\alpha \approx 1$ the number of Monte Carlo samples required in order to generate one sample satisfying the conditioning event increases like $1/(1-\alpha)$. On the contrary, the SMC sampler is constrained to always use a fixed number of particles, independently of the rareness of the event. This is a significant advantage of such an approach.

\begin{figure}
\begin{center}
\includegraphics[width=0.4\textwidth]{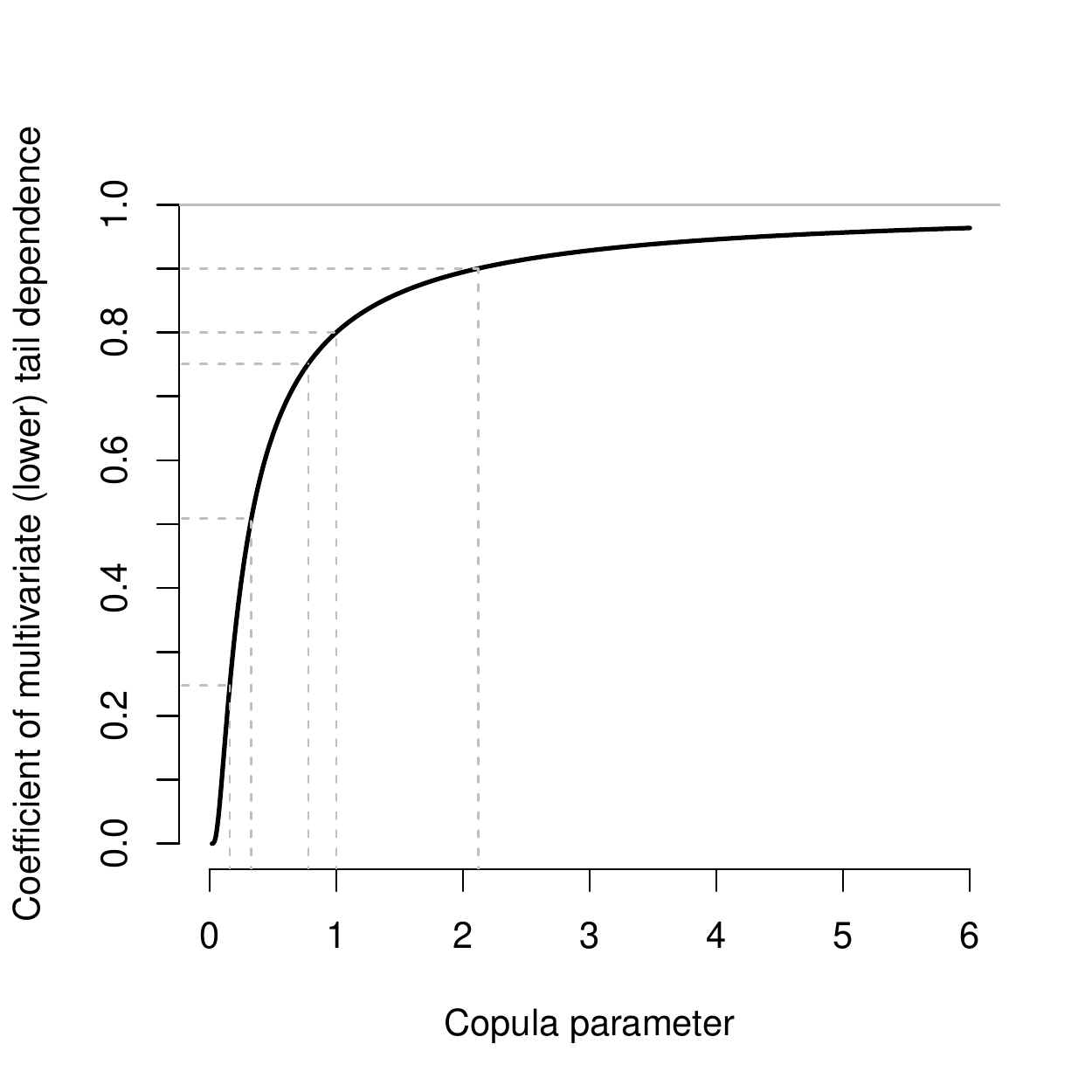}
\end{center}
\caption{Coefficient of multivariate lower tail dependence for a $5$-dimensional Clayton copula.}
\label{fig:clayton_lower_tail_dep}
\end{figure}

\begin{figure}
\begin{center}
\includegraphics[width=1\textwidth]{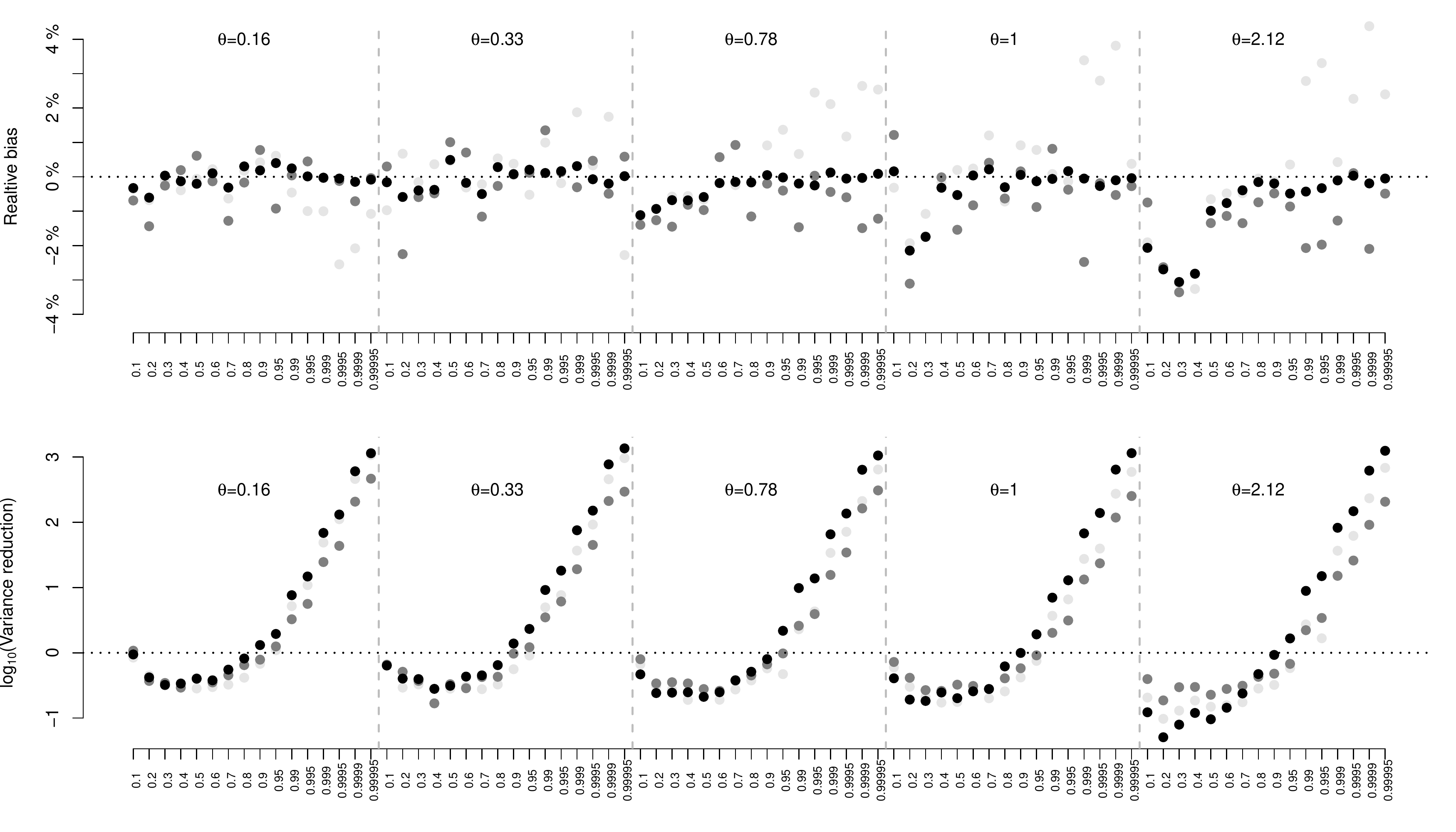}
\end{center}
\caption{Relative Bias (top) and Variance Reduction (bottom) for the $5$-dimensional Clayton copula using the SMC algorithm. Using the notation from (\ref{eq:target_exp}), {\color{light-gray}$\bullet$} Marginal for $i=1$, {\color{mid-gray}$\bullet$} Marginal for $i=5$, {\color{black}$\bullet$} Sum of all the marginal conditional expectations (Expected Shortfall).  }
\label{fig:clayton_bias_var}
\end{figure}

For the ACH Importance Sampling algorithm of \cite{arbenz2014importance}, discussed in Section \ref{sec:ACH}, we follow the suggestion proposed in the original work involving the use of a discrete version of the optimal mixing distribution $F_\Lambda$ (see (\ref{eq:f_v_arbenz})) with mass concentrated on the following 20 points: $x_k= 1-0.5^k$, $k=1,...,k$. For a fixed quantile level $\alpha$ and parameter $\theta$, the calibration of $F_\Lambda$ follows the proposed procedure in \cite{arbenz2014importance}, Section 6.1 which uses the stop-loss as the objective function
$$\widetilde{\Psi}(\textbf{u}) = \max\left\{ \sum_{i=1}^d F_i^{-1}(u_i) - VaR_{\alpha}\left(\sum_{i=1}^d F_i^{-1}(u_i) \right), \, 0 \right\}.$$

Whilst the SMC algorithm is asymptotically unbiased (although it can be seen from Figure \ref{fig:clayton_bias_var} the bias can be negligible even for finite $N_{SMC}$) the IS - ACH is unbiased for any finite sample size $N_{IS} < \infty$. Therefore it is not necessary to analyse the Relative Bias of the method. 

On the other hand, following the notation on Section \ref{sec:ACH} for a fixed parameter $\theta$ a new efficiency measure can be studied as a function of $\alpha$. We will denote by $P_{IS}(\alpha)$ the ``percentage of particles with non-zero weight'' for the $\alpha$ quantile. Formally this quantity is defined as
\begin{equation}
P_{IS}(\alpha) = \frac{\Exp[\widetilde{N}]}{\Exp[N_\textbf{V}]N_{IS} },
\label{eq:M_IS}
\end{equation}
where $N_{IS}$ is the desired sample size of the algorithm and $\widetilde{N}$ and $\Exp[N_\textbf{V}]$ are, respectively, the number of particles with non-zero weight and the number of draws in the rejection algorithm in order to have one sample from $F_\textbf{V}$ (see Section \ref{sec:ACH}). Intuitively we should expect some Variance Reduction if, and only if, the quantity $P_{IS}(\alpha)$ is larger than $1-\alpha$.

As in the analysis made for the SMC algorithm, for the IS-ACH we also look at the (rescaled) Variance Reduction. To take into account the rejection steps in the algorithm we worked with the following Variance Reduction formula
\begin{equation}
\text{Variance Reduction} = N_{MC} \times Var(\widehat{\varphi}_{MC}) \Bigg/ \Exp[N_\textbf{V}] \times N_{IS} \times Var(\widehat{\varphi}_{IS}).
\label{eq:varRed_IS}
\end{equation}

\begin{figure}
\begin{center}
\includegraphics[width=1\textwidth]{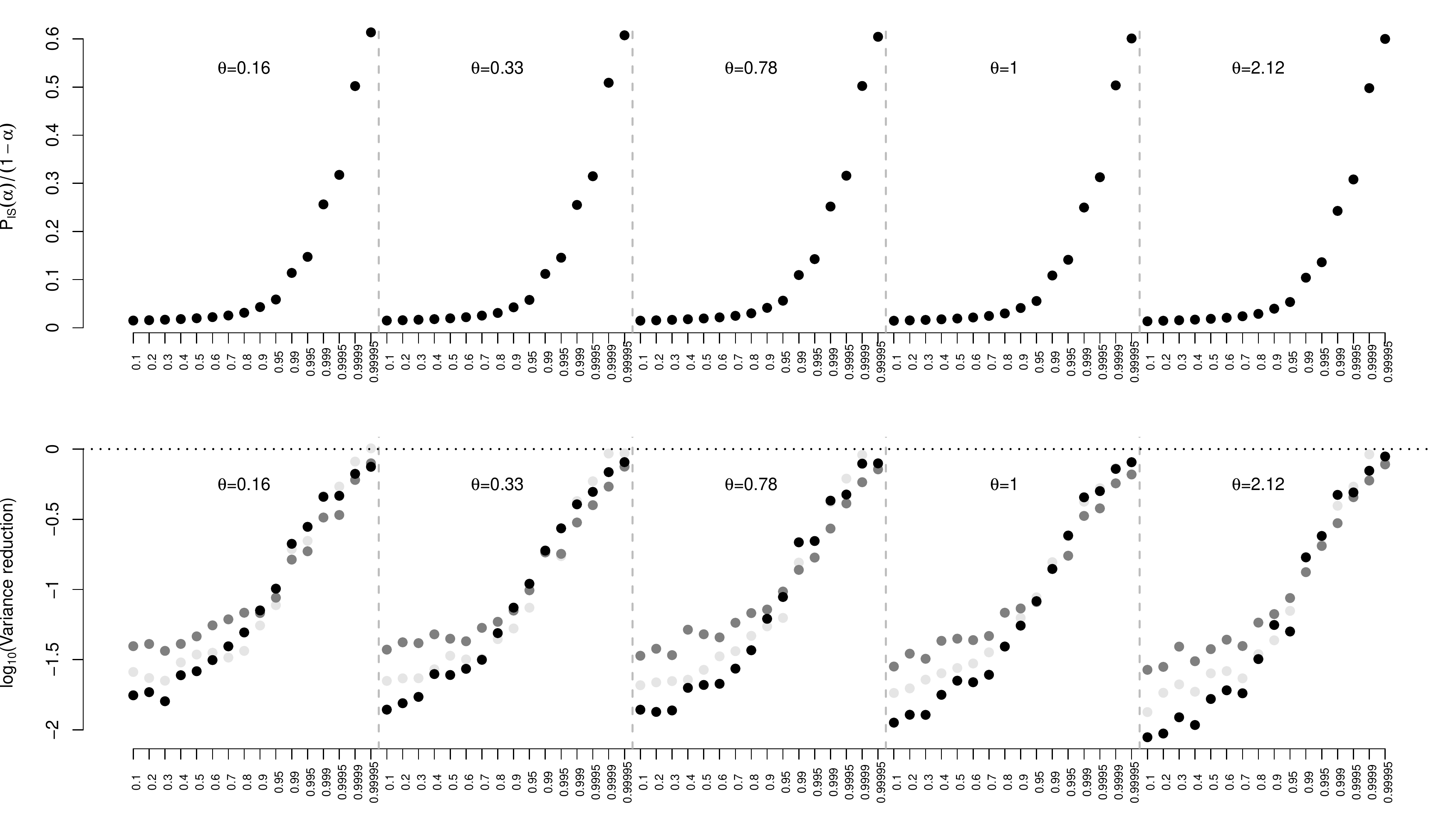}
\end{center}
\caption{Ratio between the percentage of particles with non-zero weight and $1-\alpha$ (top) and Variance Reduction (bottom) for the $5$-dimensional Clayton copula using the IS - ACH algorithm. Using the notation from (\ref{eq:target_exp}), {\color{light-gray}$\bullet$} Marginal for $i=1$, {\color{mid-gray}$\bullet$} Marginal for $i=5$, {\color{black}$\bullet$} Sum of all the marginal conditional expectations (Expected Shortfall).  }
\label{fig:clayton_IS_eff_var}
\end{figure}

From Figure \ref{fig:clayton_IS_eff_var} (top) we can see that the percentage of particles with non-zero weight, $P_{IS}(\alpha)$, is always smaller than the $1-\alpha$, indicating an inefficiency of the IS - ACH. This inefficiency is verified in the bottom of the same figure, where the scaled Variance Reduction (as of \ref{eq:varRed_IS}) is presented. As in the SMC case, the Variance Reduction factor increases as a function of $\alpha$, but in the IS-ACH case it is always smaller than 1. Although this is the case, we can expect the method to be efficient (in the Variance Reduction sense) as we get even closer to $\alpha = 1$.

%===========================================================================
%===========================================================================
%===========================================================================
\subsection{\textbf{Gumbel copula dependence between risk cells}}
In the second example we analyse the impact of the dimension in the estimation of conditional expectations when the copula is assumed to be from the Gumbel family (see Definition \ref{def:arch}). In this case we have chosen one parameter value $(\theta = 1.25)$ in order to have values for the coefficient of multivariate upper tail dependence ($\lambda_u$ from \cite{de2012multivariate} and Definition \ref{def:mult_tail_dep}) ranging from very mild dependence $(\lambda_u \approx 0.25)$ up to very strong dependence $(\lambda_u \approx 0.9)$ in a highly constrained copula density, ie, a single Gumbel copula in up to 7 dimensions.

The SMC algorithm was studied for examples including dimensions $d=2,3,...,7$ with $N_{SMC}=250$ particles and the results are presented on Figure \ref{fig:gumbel_bias_var}. From the top row we can see the Relative Bias of the conditional expectations for low dimensional copulas (eg, $d=2$ or $d=3$) is well behaved, being at most $5\%$ of the true (Monte Carlo) value when $d=2$ for all quantiles, but when the dimensionality of the problem increases a larger bias is observed. This is expected, as a single Gumbel copula in 7 dimensions, for instance, is highly constrained and its mass is mostly concentrated in a small area of the upper right quadrant of the 7-dimensional hypercube $[0,1]^7$. In the worst case, for the first marginal of a $d$ dimensional copula the Relative Bias reaches more than $40\%$ of the true value. To reduce this bias, one must increase the number of particles in the SMC Sampler. Here we have selected a very conservative set of $N_{SMC}=250$ particles. Next we studied the bias reduction as the number of particles increases, verifying the asymptotic unbiasedness of the SMC Sampler when $N_{SMC} \rightarrow \infty$. 

From the bottom row of Figure \ref{fig:gumbel_bias_var} we can see that in all dimensions presented the SMC method is highly effective regarding decreasing the variance of the estimates when the quantile is larger than $0.999$ but, as in the Clayton case, it is less efficient than a simple MC when the quantile is low.

\begin{figure}
\begin{center}
\includegraphics[width=0.4\textwidth]{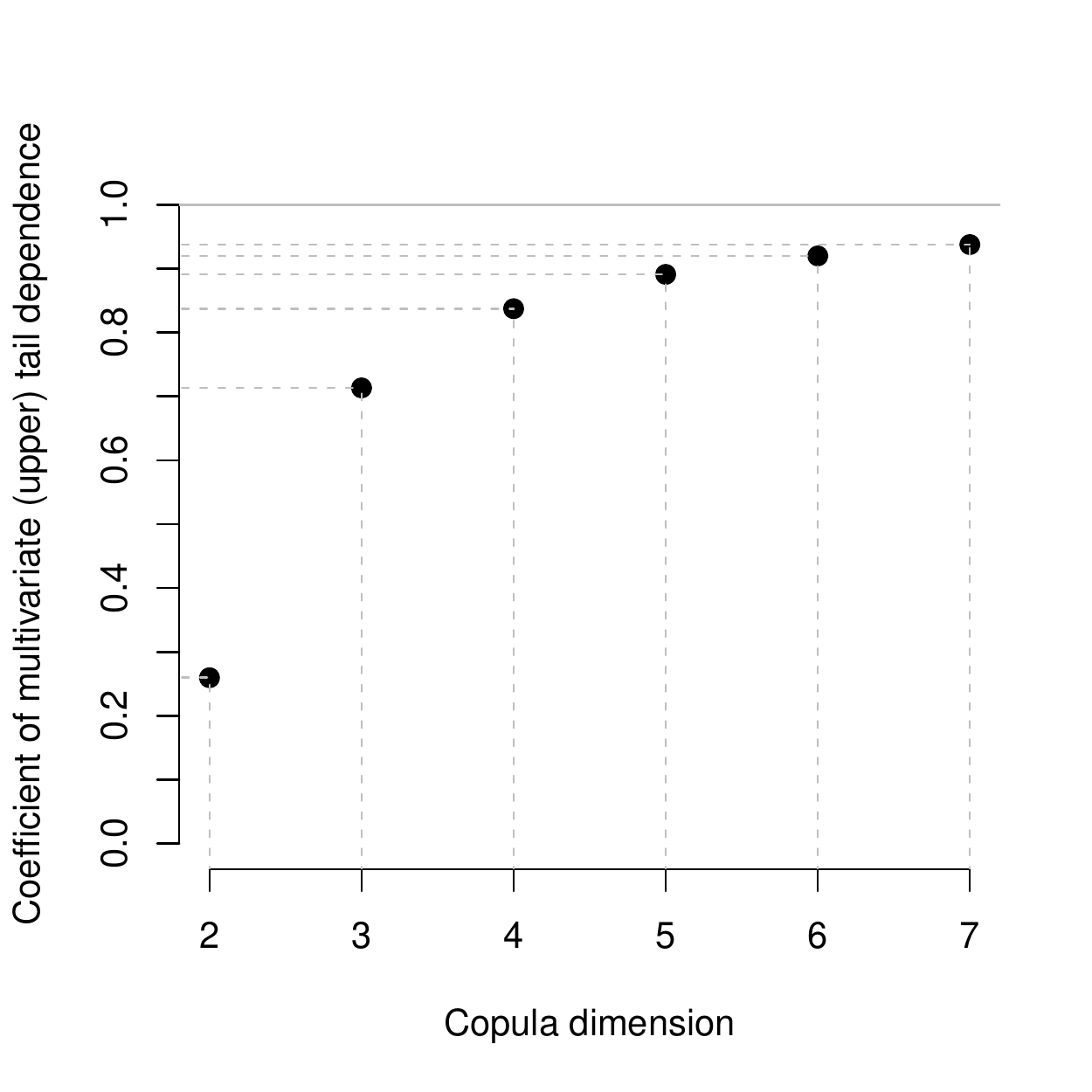}
\end{center}
\caption{Coefficient of multivariate upper tail dependence for a Gumbel(1.25) copula.}
\label{fig:gumbel_upper_tail_dep}
\end{figure}

\begin{figure}
\begin{center}
\includegraphics[width=1\textwidth]{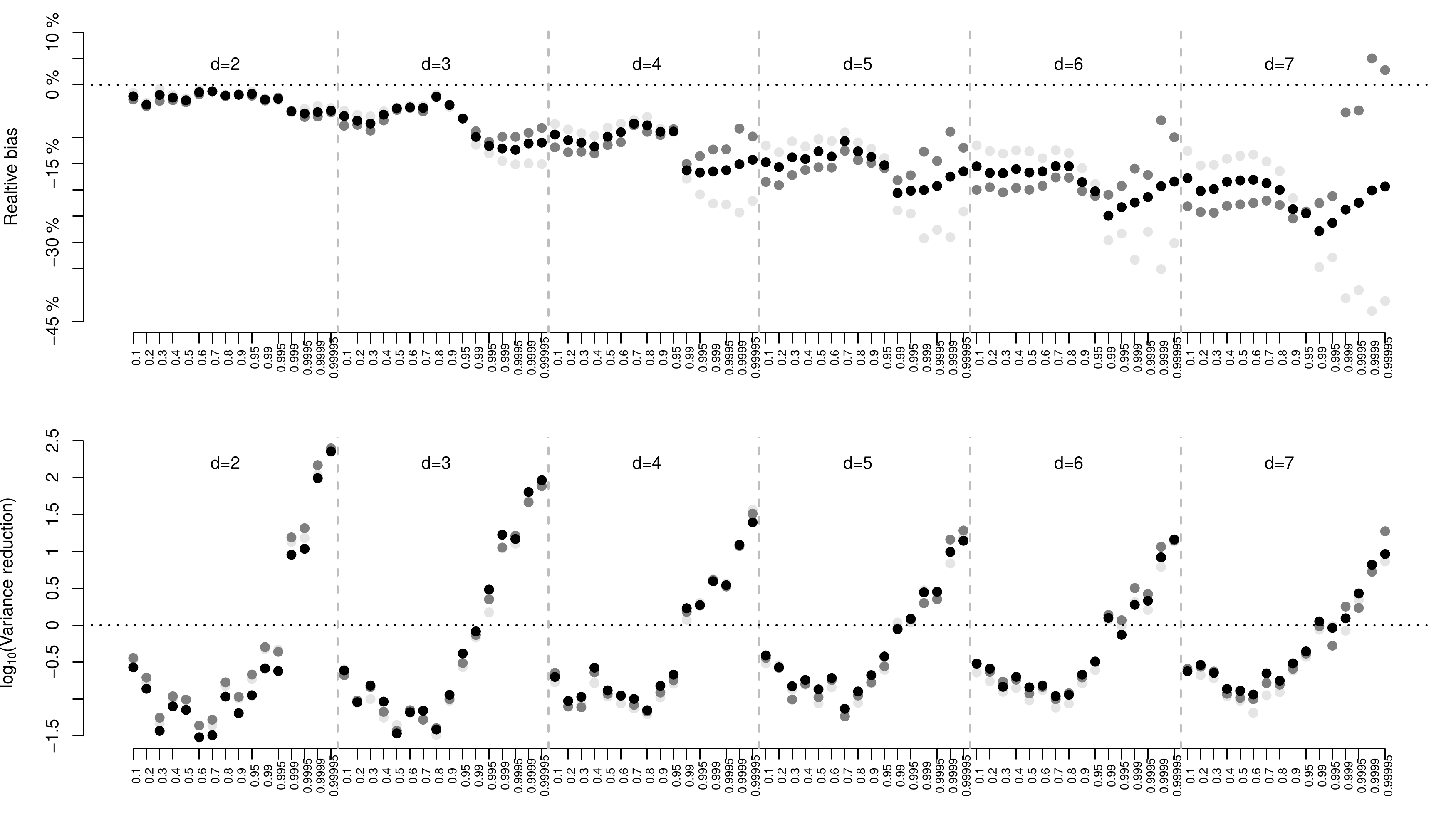}
\end{center}
\caption{Relative Bias (top) and Variance Reduction (bottom) for the Gumbel(1.25) copula using the SMC algorithm. Using the notation from (\ref{eq:target_exp}), {\color{light-gray}$\bullet$} Marginal for $i=1$, {\color{mid-gray}$\bullet$} Marginal for $i=d$, {\color{black}$\bullet$} Sum of all the marginal conditional expectations (Expected Shortfall).  }
\label{fig:gumbel_bias_var}
\end{figure}

Even though the bias involved in the SMC procedure is large for dimensions larger than $d=3$, Figure \ref{fig:gumbel_bias_var_large_N} shows that one can decrease the absolute bias by increasing the sample size used in the SMC algorithm. For example, for the first marginal in $6$ dimensions the Relative Bias goes from $-35\%$ to $-30\%$ when the number of particles increases from $N_{SMC}=250$ to $N_{SMC}=1,000$. The drawback of the increase in the sample size is that the method gets less effective in the Variance Reduction sense, although even in the case where $d=6$ and $N_{SMC}=1,000$ we still observe some humble improvement in the variance.

\begin{figure}
\begin{center}
\includegraphics[width=1\textwidth]{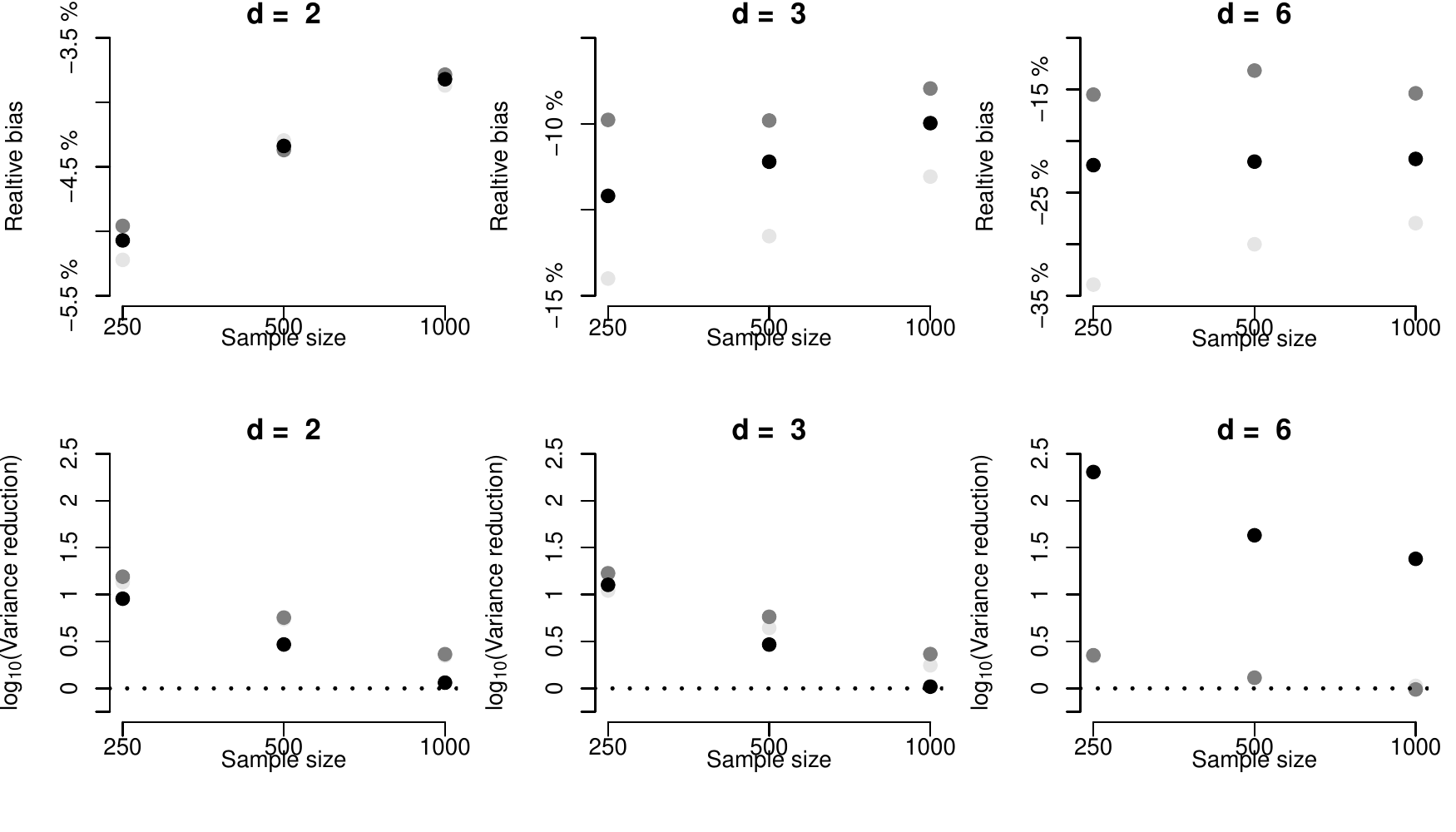}
\end{center}
\caption{Relative Bias (top) and Variance Reduction (bottom) for the Gumbel(1.25) copula using the SMC algorithm with $N_{SMC} = 250, \, 500$ and $1,000$ particles. Using the notation from (\ref{eq:target_exp}), $\alpha = 0.999$, {\color{light-gray}$\bullet$} Marginal for $i=1$, {\color{mid-gray}$\bullet$} Marginal for $i=d$, {\color{black}$\bullet$} Sum of all the marginal conditional expectations (Expected Shortfall).  }
\label{fig:gumbel_bias_var_large_N}
\end{figure}

It is important to note that the estimation of expectations of the form (\ref{eq:target_exp}) in the Gumbel model is extremely challenging, specially due to the fact that, differently from the Clayton copula, the Gumbel copula possess an intricate dependence structure near the upper right corner of the unit cube. In this case the exploration of the $[0,1]^d$ needs to be done in a (even more) careful way, in order to avoid regions with low probability density. In practice it will also be important to consider the design of the mutation kernel in the SMC Sampler algorithm if higher dimensions are considered.

%===========================================================================
%===========================================================================
%===========================================================================
\subsection{\textbf{Hierarchical Clayton copula dependence between Business Units and Event Types}}
As a final example, we show the utilization of the SMC Sampler method in a hierarchical allocation process (as in Section \ref{sec:hierarchical_allocation}). As a toy model, we assume a bank is divided in two different Business Units (B.U.). For the first one it is assumed that Operational Losses are due to three different Event Types (E.T.), while for B.U.2 losses may come from four different E.T.. For the simulation and also density calculation, in this example we made use the R-package {\ttfamily HAC} \cite{okhrin2014}.

This bank structure can be conveniently modelled with the help of a Hierarchical Archimedean Copula (HAC), also known as Nested Archimedean Copula (see \cite{okhrin2014} and references therein). For this example we have chosen a Hierarchical Archimedian Copula as in Figure \ref{fig:HAC}. The dependence of the three E.T.'s on B.U.1 is given by a Clayton copula with parameter $\theta_1 = 0.75$, while within the $4$ E.T.'s in B.U.2 the dependence is modelled through a Clayton copula with parameter $\theta_2 = 1$. Moreover, any loss on B.U.1 is related to losses in B.U.2 through a Clayton copula with parameter $\theta_0 = 0.5$. The copula for this model is given by
$$C(\ubold) = C_0(C_1(\ubold_1; \, \theta_1), \, C_2(\ubold_2; \, \theta_2); \, \theta_0),$$
where $C( \, \cdot \, ; \, \theta)$ denotes a Clayton copula with parameter $\theta$ and $\ubold = (u_1,...,u_7)$, $\ubold_1 = (u_1,u_2,u_3)$, $\ubold_2 = (u_4,u_5, u_6, u_7)$. The reader should note that $C_0( \, \cdot \, ; \, \theta_0)$ is not a copula between aggregated losses. It is also important to stress the fact that this choice of parameters will ensure the Hierarchical Copula is a well defined copula, since all the members are from the same family and the parameters are decreasing from the highest to the lowest level (see, for example, \cite{hofert2010sampling}).

\begin{figure}
\begin{center}
\includegraphics[width=0.6\textwidth]{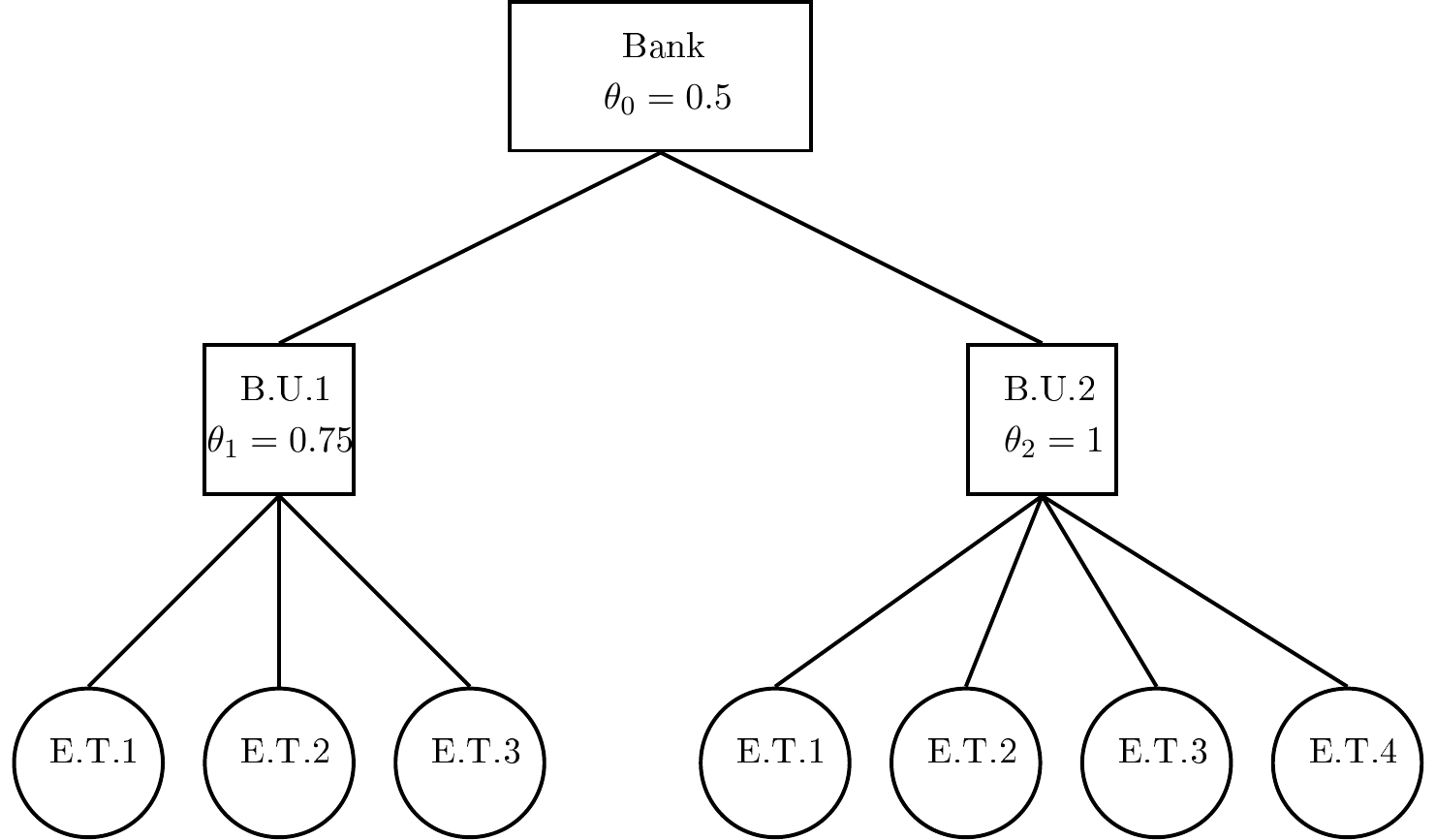}
\end{center}
\caption{Hierarchical Clayton Copula.}
\label{fig:HAC}
\end{figure}

As in the non-nested Clayton case, from Figure \ref{fig:HAC_bias_var} we can see the SMC Sampler procedure is unbiased for $N_{SMC}=250$, with Relative Bias smaller than $5\%$ in absolute terms. The method also decreases the variance of the estimates when the quantile level in the conditional event is larger than $\alpha = 0.99$, from where we can state its effectiveness.

\begin{figure}
\begin{center}
\includegraphics[width=0.9\textwidth]{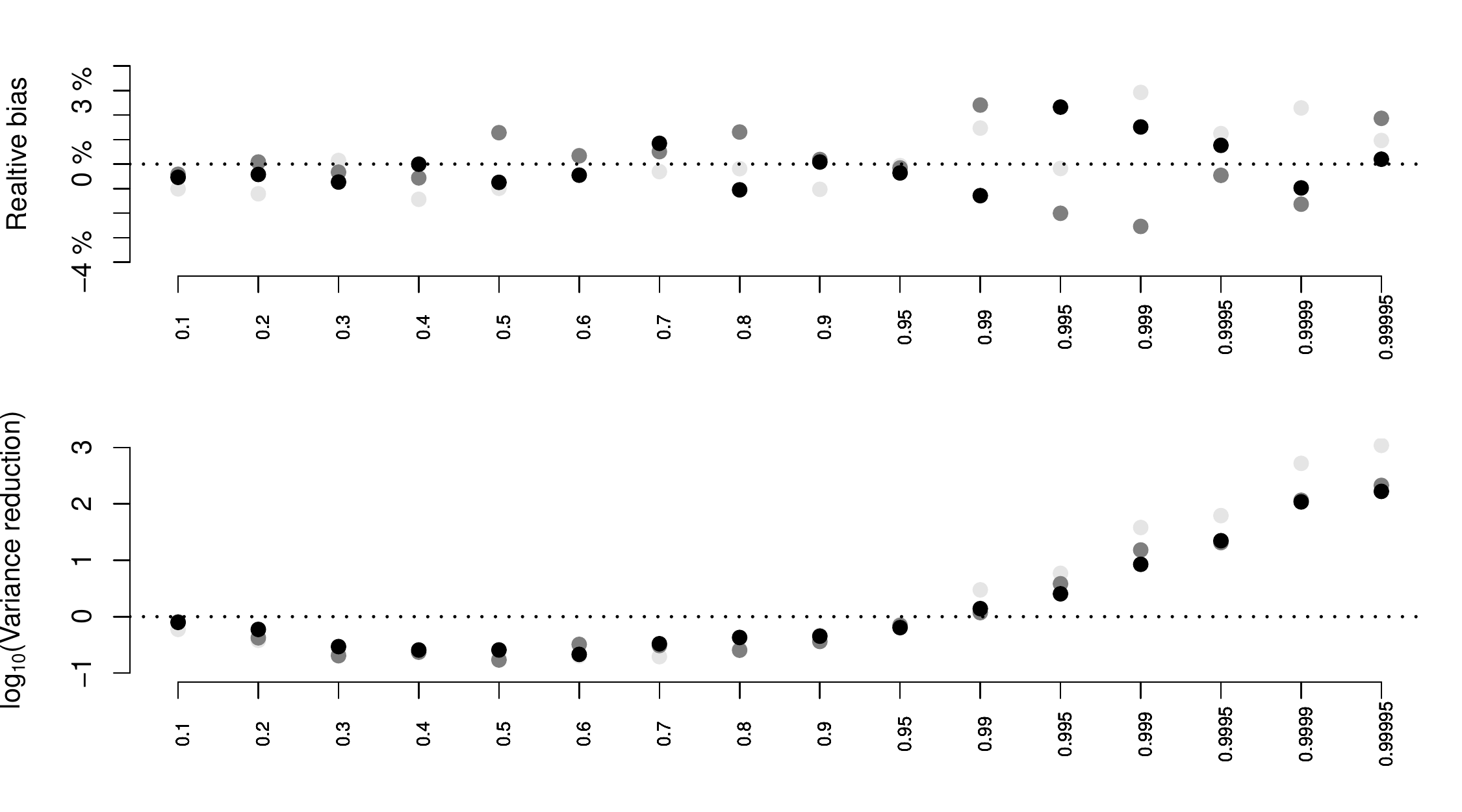}
\end{center}
\caption{Relative Bias (top) and Variance Reduction (bottom) for the Hierarchical Clayton copula from Figure \ref{fig:HAC} using the SMC algorithm with $N_{SMC} = 250$ particles. Using the notation from (\ref{eq:target_exp}), {\color{light-gray}$\bullet$} Marginal for $i=1$, {\color{mid-gray}$\bullet$} Marginal for $i=7$, {\color{black}$\bullet$} Sum of all the marginal conditional expectations (Expected Shortfall).  }
\label{fig:HAC_bias_var}
\end{figure}

%===========================================================================
%===========================================================================
%===========================================================================
\section{Conclusion and Final Remarks} \label{sec:conclusions}
With focus on the capital allocation problem for copula-dependent risks, we presented a Sequential Monte Carlo (SMC) Sampler algorithm to calculate conditional expectations where the conditioning event is rare. We exploit the copula structure to design a SMC algorithm whose efficiency is analysed through the Variance Reduction (see Equation \ref{eq:variance_reduction}) when compared to a simple Monte Carlo scheme. 

If, in addition to the Variance Reduction, the computational time is to be computed, one should be extremely careful, as it will be strongly dependent the programming language used, the hardware and also the actual implementation of the algorithms. For example, in {\ttfamily R} most of the basic functions (including sampling from simple distributions) can be used in a vectorial form, meaning that applying the function to a vector will be much faster than calculating the function values serially, ie, inside a {\ttfamily for} loop. For the sake of simplicity, both in the simple MC and in the SMC schemes we sampled one value at a time and did not make use of any vectorial form.

Due to the nature of the allocation problem, there is no need for the algorithm presented to be run \emph{online}. In most of the cases it will be used once a month or even once a year. Nevertheless, if the performance of the SMC algorithm needs to be improved specialized libraries such as the \emph{LibBi} (see \cite{murray2013bayesian}) can be used.

Although the code used for this paper has not been thoroughly optimised, some computational analysis can still be done. For the Clayton example from Section \ref{sec:clayton} (with parameter $\theta=1$), for each quantile level, Table \ref{tbl:clayton_time} presents (1) the time (in minutes) necessary to run the SMC algorithm; (2) how many times the SMC algorithm is slower than the simple Monte Carlo; (3) the Variance Reduction factor (as of (\ref{eq:variance_reduction})) for the quantities in (\ref{eq:target_exp}), where ES denotes the sum of the expectations (Expected Shortfall).

\begin{table}[ht]
\centering
\begin{tabular}{rrr|rrr}
	\hline
\rowcolor{white}
  & & & \multicolumn{3}{c}{\textbf{Variance Reduction}} \\
	\hline
\textbf{Quantile} & \textbf{SMC time} & $\displaystyle \frac{\text{\textbf{SMC time}}}{\text{\textbf{MC time}}}$ & $\textbf{k=1}$ & $\textbf{k=5}$ & \textbf{ES}\\[2ex]
  \hline
  0.1 & 0.32 & 53.83 & 0.34 & 1.95 & 0.46 \\ 
  0.2 & 0.64 & 68.97 & 0.27 & 1.16 & 0.29 \\ 
  0.3 & 0.97 & 73.94 & 0.11 & 0.14 & 0.37 \\ 
  0.4 & 1.30 & 75.04 & 0.07 & 1.89 & 0.19 \\ 
  0.5 & 1.63 & 73.09 & 0.19 & 0.16 & 0.18 \\ 
  0.6 & 1.97 & 68.15 & 0.35 & 0.03 & 0.28 \\ 
  0.7 & 2.31 & 61.41 & 0.42 & 0.21 & 0.78 \\ 
  0.8 & 2.65 & 52.55 & 0.30 & 0.38 & 0.21 \\ 
  0.9 & 3.00 & 39.28 & 0.44 & 2.94 & 1.83 \\ 
  0.95 & 3.35 & 26.34 & 0.29 & 0.33 & 4.10 \\ 
  0.99 & 3.70 & 9.64 & 2.79 & 5.36 & 3.54 \\ 
	\hline
  0.995 & 4.04 & 4.56 & 2.98 & 6.54 & 21.87 \\ 
  0.999 & 4.39 & 1.28 & 16.56 & 7.20 & 71.37 \\ 
  0.9995 & 4.74 & 0.56 & 94.74 & 88.12 & 172.51 \\ 
  0.9999 & 5.08 & 0.15 & 224.82 & 154.54 & 272.74 \\ 
  0.99995 & 5.42 & 0.06 & 1891.81 & 205.48 & 593.32 \\  
	\hline
\end{tabular}
\label{tbl:clayton_time}
\caption{Computational time (in minutes) and Variance Reduction for the SMC algorithm when compared to a simple Monte Carlo scheme.}
\end{table}

On this particular implementation of the algorithm, for low quantiles we can see that the MC scheme is considerably faster for the same accuracy of the SMC method. For example, when $\alpha= 0.4$ the MC method is 75 times faster than the SMC and the Variance Reduction on the first marginal expectation is 0.07, meaning that the MC variance is around 14 times smaller than the SMC variance.

On the other hand, for higher quantiles, in particular after $\alpha=0.99$, the SMC method starts to become more appealing, since the Variance Reduction gets larger than the difference in time. It's worth noticing the proposed SMC method has been designed to be used on extreme quantiles where lower quantiles are only used as intermediate steps.

The computing times presented on Table \ref{tbl:clayton_time} were measured using {\ttfamily R 3.1.0} in a Intel Xeon E5-1650, 3.20GHz and 16GB RAM. For each quantile level the algorithms were run 10 times and the values presented are an average over these 10 runs.

The {\ttfamily R} code used is available upon request from the authors.

\begin{acknowledgements}
The authors would like to thank Prof. Mario W\"{u}thrich (ETH, Switzerland) and Prof. Pierre Del Moral (UNSW, Australia) for fruitful discussions related to this work. The authors also acknowledge the use of the UCL Legion High Performance Computing Facility (Legion@UCL), and associated support services. RST has been funded by the CNPq (Brazil) through a Ci\^{e}ncia sem Fronteiras scholarship and part of this project has been developed while he was an Endeavour Research Fellow at CSIRO, Australia. GWP acknowledges the support of the Institute of Statistical Mathematics, Japan; the Oxford-Man Institute, UK and the Systemic Risk Center, London School of Economics, UK.
\end{acknowledgements}

%\newpage
\appendix\normalsize
\section{Euler's homogeneous function theorem}
\label{sec:eulerThm}
%In this section we provide a characterization of homogeneous functions, through Euler's theorem.

\begin{definition}[Homogeneous function] A function $f \, : \, U \subset \R^d \rightarrow \R$ is said to be homogeneous of degree $\tau$ if, for all $h >0$ and $u \in U$ with $hu \in U$ the following equation holds:
$$f(hu) = h^\tau f(u).$$
\end{definition}

\begin{theorem}[Euler's homogeneous function theorem] Let $U \subset \R^n$ be an open set and $f \, : \, U \rightarrow \R$ be a continuously differentiable function. Then $f$ is homogeneous of degree $\tau$ if, and only if, it satisfies the following equation:
$$\tau f(\textbf{u}) = \sum_{i=1}^n u_i \frac{\partial f}{\partial u_i}, \ \ \ \textbf{u}=(u_1,...,u_n) \in U, \, h>0$$

\end{theorem}

%================================================================================================
%================================================================================================
%================================================================================================
\section{Copulas and Sklar's theorem}
\label{sec:gaussianCopula}

\begin{definition}[Copula] A d-dimensional copula is a distribution function on $[0,1]^d$ with uniform marginal distributions.
\end{definition}

\begin{theorem}[Sklar] Let $F_\X$ be a joint distributions with margins $F_1,...,F_d$. Then there exists a copula $C:[0,1]^d \rightarrow [0,1]$  such that 
\begin{equation}\label{copula}
	F_\X ( \x ) = C\big( F_1(x_1),...,F_d(x_d)\big), \ \ \forall \x =(x_1,...,x_d) \in \overline{\R}^d.
\end{equation}
If the margins are continuous then $C$ is unique, given by
$$C(u_1,...,u_d) = F_\X(F_1^{-1}(u_1), ..., F_d^{-1}(u_d))$$
Conversely, if $C$ is a copula and $F_1,...,F_d$ are univariate distributions, then the $F$ defined in (\ref{copula}) is a joint distribution function with margins $F_1,...,F_d$.
\end{theorem}

Moreover, if we assume that $F_1,...,F_d$ are differentiable, then the joint density function of $\X$ can be written as
$$f_\X(\x) = c\big( F_1(x_1),...,F_d(x_d)\big) \prod_{i=1}^d{f_i(x_i)},$$
where
$$\displaystyle c(u_1,...,u_d) = \frac{\partial^d C(u_1,...,u_d)}{\partial u_1 ... \partial u_d}$$
and $f_i$ is the density of $X_i$.

\begin{definition}[Archimedean generator] An Archimedean generator is a continuous, decreasing function $\psi:[0, \infty] \rightarrow [0,1]$ that satisfies $\psi(0) = 1, \ lim_{t \rightarrow \infty} \psi(t) =0$ and is strictly decreasing on $[0, \ \inf\{t \ : \ \psi(t) =0 \}]$.
\end{definition}

\begin{definition}[Archimedean copulas]  A $d$-dimensional copula is called Archimedean if it is of
\begin{equation}
C(\ubold; \ \psi) = \psi(\psi^{-1}(u_1) +  \psi^{-1}(u_d)), \quad \ubold = (u_1,...,u_d) \in [0,1]^d,
\label{eq:archimedean}
\end{equation}
where $\psi$ is the Archimedean generator.
\label{def:arch}
\end{definition}

\begin{table}[ht]
	\centering
		\begin{tabular}{lll}
\hline
\textbf{Family} & \textbf{Parameter} & \textbf{Generator} $\psi(t)$ \\
\hline
\rowcolor{white}
Clayton & $\theta \in (0,\infty)$ & $(1+t)^{-1/\theta}$ \\
Gumbel & $\theta \in [1,\infty)$ & $\exp{-t^{1/\theta}}$ \\
\hline
		\end{tabular}
	\caption{Commonly used Archimedean generators}
	\label{tab:generators}
\end{table}

\begin{definition}[Multivariate coefficients of tail dependence] \label{def:mult_tail_dep} For the copulas defined in Table \ref{tab:generators} the multivariate coefficients of upper and lower tail dependence are defined, respectively, as
\begin{align*}
\lambda_u &= \lim_{u\rightarrow 1^-} \Prob[U_1 > u \, | \, U_2 > u,...,U_d>u] \\
					&= \lim_{t\rightarrow 0^+} \frac{\sum_{i=1}^d { n \choose n-i} i (-1)^i \psi'(it)}{\sum_{i=1}^{n-1}{ n -1 \choose n-1-i} i (-1)^i \psi'(it) }
%\label{eq:upper_tail_dep}
\end{align*}

\begin{align*}
\lambda_l & = \lim_{u\rightarrow 0^+} \Prob[U_1 \leq u \, | \, U_2 \leq u,...,U_d\leq u] \\
          & = \frac{d}{d-1}\lim_{t \rightarrow \infty} \frac{\psi'(d t)}{\psi'\big( (d-1)t\big)}
%\label{eq:lower_tail_dep}
\end{align*}

\end{definition}

% sample references
% %
% Use this file as a template for your own input.
%
%%%%%%%%%%%%%%%%%%%%%%%% Springer-Verlag %%%%%%%%%%%%%%%%%%%%%%%%%%
%
% BibTeX users please use
 \bibliographystyle{spmpsci}
 \bibliography{bib_CapitalAllocation}

\end{document}